\documentclass[10pt,a4paper]{article}

\usepackage{amsmath,amssymb,amsfonts}
\usepackage{amsthm}
\usepackage[a4paper,margin=2.5cm]{geometry}
\usepackage[numbers,sort&compress]{natbib}
\usepackage[colorlinks=true, citecolor=blue, urlcolor=blue]{hyperref}

\allowdisplaybreaks

\DeclareMathOperator{\Span}{span}
\DeclareMathOperator{\Tr}{Tr}
\DeclareMathOperator{\id}{id}
\DeclareMathOperator{\diag}{diag}
\DeclareMathOperator{\rank}{rank}

\newcommand{\CPTP}{\mathrm{CPTP}}
\newcommand{\HPTP}{\mathrm{HPTP}}
\newcommand{\HS}{\mathrm{HS}}
\newcommand{\cA}{\mathcal{A}}
\newcommand{\cB}{\mathcal{B}}
\newcommand{\cD}{\mathcal{D}}
\newcommand{\cE}{\mathcal{E}}
\newcommand{\cH}{\mathcal{H}}

\newcommand{\cK}{\mathcal{K}}
\newcommand{\cL}{\mathcal{L}}
\newcommand{\cM}{\mathcal{M}}
\newcommand{\cN}{\mathcal{N}}

\newcommand{\cR}{\mathcal{R}}

\newcommand{\cV}{\mathcal{V}}
\newcommand{\cW}{\mathcal{W}}
\newcommand{\cZ}{\mathcal{Z}}
\newcommand{\ket}[1]{|#1\rangle}
\newcommand{\bra}[1]{\langle#1|}
\newcommand{\braket}[2]{\langle #1|#2\rangle}
\newcommand{\kett}[1]{|#1\rangle\!\rangle}
\newcommand{\bratt}[1]{\langle\!\langle #1|}

\theoremstyle{plain}
\newtheorem{theorem}{Theorem}
\newtheorem{lemma}{Lemma}
\newtheorem{proposition}{Proposition}

\theoremstyle{definition}
\newtheorem{definition}{Definition}
\newtheorem{example}{Example}

\theoremstyle{remark}
\newtheorem{remark}{Remark}

\renewenvironment{proof}[1][\proofname]%
  {\par\noindent\emph{#1.}\ }{\hfill$\blacksquare$\par}

\title{\textbf{Approximate Virtual Quantum Recovery}}

\author{Yuan Liu\thanks{MIIT Key Laboratory of Complex-field Intelligent Exploration, School of Optics and Photonics, Beijing Institute of Technology, Beijing 100081, China.}%
\and
Linhan Lin\thanks{State Key Laboratory of Precision Measurement Technology and Instruments, Department of Precision Instrument, Tsinghua University, Beijing 100084, China. E-mail: linlh2019@mail.tsinghua.edu.cn.}%
\and
Ke-Mi Xu\thanks{MIIT Key Laboratory of Complex-field Intelligent Exploration, School of Optics and Photonics, Beijing Institute of Technology, Beijing 100081, China.
E-mail: xukemi@bit.edu.cn.}}

\date{}

\begin{document}

\maketitle

\begin{abstract}
The recovery of quantum information after subsystem loss is a central challenge in quantum information processing. However, some quantum states remain beyond the reach of any recovery strategies. Here we identify the algebraic origin of virtual irrecoverability, the \emph{ghost information}---correlations present in the global state that leave no trace on any accessible subsystem. We introduce a scalar measure to quantify its magnitude and prove a universal error floor below which no virtual recovery map can operate, irrespective of resource investment. In such cases only \emph{approximate} virtual recovery is available. We derive upper and lower bounds on the sampling overhead for approximate virtual recovery. Moreover, we show that conditional mutual information, the standard entropic diagnostic, does not determine the sampling overhead. Our work thus establish an algebraic framework of information recovery beyond entropic criteria.
\end{abstract}

\section{Introduction}

A tripartite quantum state $\rho_{ABC}$ encodes correlations among $A$, $B$, and $C$ that its bipartite reductions can, in general, only partially capture \cite{Amico08RMP, Horodecki09RMP, Qiu26TIT}. When a subsystem becomes inaccessible \cite{Pirandola18NPhoton}, as occurs routinely in noisy quantum sensors \cite{Yamamoto22Dec, Demkowicz12NC}, photon-limited imaging \cite{Genovese16JOpt, Nair11Oct}, and lossy quantum networks \cite{Zhang21QST, Nehra24arXiv}, the reduced state usually carries less information than the original state. Whether this information loss can be reversed by a recovery operation acting solely on the surviving subsystem is a question of both fundamental and practical importance.

For quantum channels, i.e., completely positive and trace-preserving (CPTP) maps \cite{Watrous18book, Wilde17book}, Hayden et al.\ \cite{Hayden04CMP} proved that exact reconstruction is possible if and only if $\rho_{ABC}$ forms a quantum Markov chain (QMC), a condition equivalently characterized by vanishing conditional mutual information (CMI). For every state outside this narrow class, exact recovery fails, but approximate recovery remains possible, and its error is also controlled by the same entropic quantity. The Fawzi--Renner inequality bounds the error of the best CPTP recovery by CMI \cite{Fawzi15CMP, Sutter16PRSA, Berta16TIT}. Remarkably, if one allows recovery maps beyond CPTP maps, i.e., allows Hermitian-preserving and trace-preserving (HPTP) maps, exact reconstruction becomes possible for a much larger class of states \cite{Chen25TIT, Chen25PRR, Cai23RMP}. Such ``virtual'' maps cannot be implemented directly in the laboratory, but can be simulated: by randomly switching between two real quantum operations and weighting their measurement outcomes with positive and negative signs \cite{Zhao25PRR, Takagi22npjQuantumInf}---a technique known as quasi-probability sampling \cite{Piveteau24TIT, Pashayan15Aug, Temme17Nov, Piveteau22npjQuantumInf, Jiang21Quantum, Zhao23Quantum}---one effectively realizes an otherwise forbidden recovery. 

The states that admit exact virtual recovery are known as virtual quantum Markov chains (VQMCs) \cite{Chen25TIT}. A minimal example illustrates the distinction. The three-qubit $W$ state can be perfectly recovered by a virtual map at finite sampling overhead; the Greenberger-Horne-Zeilinger (GHZ) state, in contrast, cannot, no matter how much sampling overhead is invested. Intuitively, the GHZ state carries three-body coherence that produces a nonzero signature on the full $ABC$ system yet vanishes identically on the accessible $AB$ marginal \cite{Dur00PRA}. These results concern exact recovery only. Approximate virtual recovery, where a virtual map may tolerate a small error at a smaller sampling overhead, has not been studied. 

In this work, we systematically study approximate virtual recovery. We identify the algebraic origin of virtual irrecoverability as ghost information, i.e., correlations present in the global state but not on any subsystem that a recovery map could exploit. We introduce a scalar measure $\delta(\rho)$ to quantify its magnitude and establish a universal error floor below which no virtual recovery map can operate, irrespective of resource investment. This extends the exact-recovery classification of previous works \cite{Chen25TIT,Chen25PRR} to the approximate regime. We then turn to the sampling overhead of approximate virtual recovery and derive upper and lower bounds of it in terms of the singular spectrum of the restricted partial trace. Furthermore, we reveal that CMI, the standard entropic diagnostic, does not determine the sampling overhead (not even bounds on the sampling overhead)---a conclusion that runs counter to the numerical speculation in Ref.~\cite{Chen25TIT}, where it was conjectured that CMI might at least lower-bound the sampling overhead for VQMCs. Table~\ref{tab:prior} locates the present work relative to previous results.

\begin{table}[t]
\caption{Exact and approximate (virtual) quantum recovery. $\delta$ is the kernel failure measure, $\varepsilon_0$ the constructive error floor, and $\nu_\varepsilon$ the $\varepsilon$-approximate virtual Markovianity.}
\label{tab:prior}
\centering
\footnotesize
\begin{tabular}{@{}p{1.2cm}p{2cm}p{4cm}p{4cm}@{}}
\hline
Ref. & Setting & Established & Not available there \\
\hline
\cite{Hayden04CMP} & CPTP, exact & QMC $\Leftrightarrow$ $\mathrm{CMI}=0$ & any nonzero error \\
\cite{Fawzi15CMP, Sutter16PRSA} & CPTP, approx. & error $\lesssim\sqrt{\mathrm{CMI}}$ & HPTP maps; overhead \\
\cite{Chen25TIT} & HPTP, exact & VQMC classification; SDP for $\nu_0$ & nonzero error; relation to spectrum \\
\cite{Chen25PRR} & HPTP, exact & resource-theoretic reading & quantitative overhead \\
Here & HPTP, approx. & $\delta$, $\varepsilon_0$, $\nu_\varepsilon$ & --- \\
\hline
\end{tabular}
\end{table}

The paper is organized as follows. Section~\ref{sec:pre} fixes notation and reviews Schatten norms, the Choi--Jamio\l{}kowski isomorphism, and quasi-probability decompositions. Section~\ref{sec:floor} quantifies ghost information through the kernel failure measure and proves a universal error floor for linear recovery maps. Section~\ref{sec:ach} gives a constructive HPTP recovery map, and Section~\ref{sec:overhead} bounds the sampling overhead of approximate virtual recovery from both sides. Section~\ref{sec:CMI} shows that the conditional mutual information does not determine this overhead. We conclude in Section~\ref{sec:conclusion}. 

\section{Preliminaries}\label{sec:pre}

\subsection{Notation and Norms}

All Hilbert spaces are finite dimensional: $\cH_A=\mathbb{C}^{d_A}$, $\cH_B=\mathbb{C}^{d_B}$, $\cH_C=\mathbb{C}^{d_C}$. We write $\cL(\cH)$ for the linear operators on $\cH$. For $X\in\cL(\cH)$ we use the Schatten norms $\|X\|_1=\Tr\sqrt{X^\dagger X}$ (trace norm), $\|X\|_2=\sqrt{\Tr[X^\dagger X]}$ (Hilbert--Schmidt, HS) and $\|X\|_\infty=\max_{i}(\sum_jX_{ij})$ (operator norm), with \begin{equation}\label{eq:norm_hierarchy}
  \|\cdot\|_\infty\le\|\cdot\|_2\le\|\cdot\|_1 .
\end{equation} 

\begin{lemma}[Rank-one norm]\label{lem:rank1}
For $A=\ket{u}\bra{v}$, $\|A\|_1=\|\ket{u}\|_2\|\ket{v}\|_2$.
\end{lemma}

\begin{proof}
$A^\dagger A=(|u\rangle\langle v|)^\dagger (|u\rangle\langle v|)=\| |u\rangle\|_2^2 |v\rangle\langle v|$, which is a (unnormalized) projection operator with a single nonzero eigenvalue $\| |u\rangle\|_2^2\cdot \| |v\rangle\|_2^2$; consequently $\sqrt{A^\dagger A}$ has a single nonzero eigenvalue $\| |u\rangle\|_2\cdot \| |v\rangle\|_2$. This completes the proof, since $\|A\|_1=\Tr[\sqrt{A^\dagger A}]$, i.e., $\|A\|_1$ equals the summation of eigenvalues of $\sqrt{A^\dagger A}$.
\end{proof}

The trace norm admits the variational representation
\begin{equation}\label{eq:trace_dual}
  \|X\|_1=\max_{\|W\|_\infty\le 1}\bigl|\Tr[WX]\bigr|,
\end{equation}
where $W$ may be taken Hermitian without loss of generality. We use H\"older's inequality $|\Tr[A^\dagger B]|\le\|A\|_p\|B\|_q$ for $1/p+1/q=1$, in particular $\|AB\|_1\le\|A\|_\infty\|B\|_1$, and the rank--norm inequality
\begin{equation}\label{eq:rank_norm}
  \|X\|_1\le\sqrt{\rank(X)}\,\|X\|_2
\end{equation}
that follows from Cauchy--Schwarz applied to the singular values.

For a bipartite system $XY$, $\Tr_Y:\cL(\cH_X\otimes\cH_Y)\to\cL(\cH_X)$ is the unique linear map with $\Tr_Y[A\otimes B]=A\Tr[B]$. It is trace-norm contractive,
\begin{equation}\label{eq:ptrace_contract}
  \|\Tr_Y[X]\|_1\le\|X\|_1.
\end{equation}
On the HS level, $\Tr_C$ is the adjoint of the embedding $\cL(\cH_B)\hookrightarrow\cL(\cH_B\otimes\cH_C)$, $X\mapsto X\otimes I_C/d_C$. This identification is used implicitly in several norm estimates in this work.

\subsection{Vectorization and Choi--Jamio\l{}kowski Isomorphism}

For $X\in\cL(\cH)$, the vectorization \cite{Watrous18book} $\kett{X}\in\cH\otimes\cH$ is defined by $\kett{i}\otimes\kett{j}\leftrightarrow\ket{i}\bra{j}$, with inner product $\bratt{X}Y\rangle\!\rangle=\Tr[X^\dagger Y]$ and norm $\|\kett{X}\|_2=\|X\|_2$. Explicitly, we have $\kett{X}=X_{ij}\ket{i}\otimes\ket{j}$, $X_{ij}=\bra{i}X\ket{j}$. We can also define $\operatorname{flatten}(X)$ by $[\operatorname{flatten}(X)]_{ij}=X_{ij}$, which extracts the coefficient of $X$. Under vectorization we have
\begin{equation}\label{eq:vec_comp}
  \kett{AXB^\mathrm{T}}=(A\otimes B)\kett{X},
\end{equation}
where the transpose is taken in the computational basis. For a linear map $\cN:\cL(\cH_B)\to\cL(\cH_B\otimes\cH_C)$, the Choi matrix \cite{Watrous18book, Wilde17book} is
\begin{equation}\label{eq:Choi_def}
  J(\cN)=\sum_{i,j=0}^{d_B-1}\ket{i}\bra{j}_B\otimes\cN\bigl(\ket{i}\bra{j}_B\bigr)\in\cL(\cH_B\otimes\cH_B\otimes\cH_C).
\end{equation}

In the vectorized (Liouville) picture the map acts as $\hat\cN\kett{X}=\kett{\cN(X)}$, where $\hat\cN$ is the Liouville superoperator on $\cL(\cH_B)\to\cL(\cH_B\otimes\cH_C)$. $J(\cN)$ is related to $\hat\cN$ by a reshuffling of tensor factors. The reshuffling preserves the HS norm but not the trace norm. Specifically, for a map of the form $\Phi(X)=\Tr[B^\dagger X]A$, a direct computation gives
\begin{equation}\label{eq:traceCJ}
  J(\Phi)=\bar B\otimes A,\qquad \hat\Phi=\kett{A}\bratt{B},
\end{equation}
so that for $\Phi(X)=\Tr[X]A$ we have
\begin{equation}\label{eq:traceCJ2}
  J(\Phi)=I\otimes A,\qquad \hat\Phi=\kett{A}\bratt{I}.
\end{equation}

\subsection{HPTP Maps and Quasi-probability Decomposition}

A linear map $\cR:\cL(\cH_B)\to\cL(\cH_B\otimes\cH_C)$ is Hermitian preserving (HP) if $\cR(X)^\dagger=\cR(X^\dagger)$ for all $X$, and trace preserving (TP) if $\Tr_{BC}[\cR(X)]=\Tr[X]$. We write $\HPTP(B,BC)$ and $\CPTP(B,BC)$ for the sets of HPTP and CPTP maps from $\cH_B$ to $\cH_B\otimes\cH_{C}$, respectively. Every Hermitian operator $H$ admits a unique Jordan decomposition \cite{Watrous18book} $H=H_+-H_-$ with $H_\pm\ge0$, $H_+H_-=0$, and $\|H\|_1=\Tr[H_+]+\Tr[H_-]$.

An HPTP map admits a quasi-probability decomposition (QPD) $\cR=c_1\cN_1-c_2\cN_2$ with $c_i\ge0$ and $\cN_i\in\CPTP(B,BC)$. The optimal sampling overhead is $c(\cR)=\min\{c_1+c_2\}$ over all valid QPDs \cite{Jiang21Quantum, Piveteau24TIT}. The value $2^{c(\cR)}$ is the number of experimental runs needed per effective sample.

\begin{lemma}[Choi-norm lower bound on the overhead]\label{lem:qpd_lower}
For any HPTP map $\cR:\cL(\cH_B)\to\cL(\cH_B\otimes\cH_C)$,
\begin{equation}\label{eq:choi_qpd}
  c(\cR)\ge\frac{\|J(\cR)\|_1}{d_B}.
\end{equation}
For a completely positive (CP) map the inequality is saturated, since $\|J(\cN)\|_1=\Tr[J(\cN)]$.
\end{lemma}

\begin{proof}
For any QPD $\cR=c_1\cN_1-c_2\cN_2$ we have $J(\cR)=c_1J(\cN_1)-c_2J(\cN_2)$. Since $\cN_i$ is CPTP, $J(\cN_i)\ge0$ and $\|J(\cN_i)\|_1=\Tr[J(\cN_i)]=d_B$. Hence $\|J(\cR)\|_1\le(c_1+c_2)d_B$. Minimizing over QPDs gives \eqref{eq:choi_qpd}.
\end{proof}

\section{Ghost Information and The Error Floor}\label{sec:floor}

\subsection{Block Structure and the Virtual Markov Condition}

We briefly recall the VQMC framework \cite{Chen25TIT}. For a tripartite quantum state $\rho_{ABC}\in\cL({\cH_A\otimes\cH_B\otimes\cH_C})$, expand it in an orthonormal basis $\{\ket{i}_A\}$ of $\cH_A$,
\begin{equation}\label{eq:expansion}
  \rho_{ABC}=\sum_{i,j=0}^{d_A-1}\ket{i}\bra{j}_A\otimes Q_{BC}^{(ij)},
\end{equation}
with $Q_{BC}^{(ij)}=\bra{i}\rho_{ABC}\ket{j}_A\in\cL(\cH_B\otimes\cH_C)$ and reduced blocks $Q_B^{(ij)}=\Tr_C[Q_{BC}^{(ij)}] \in \cL(\cH_B)$. For $c=[c_{00}, c_{01}, \cdots]^\mathrm{T}$, define the linear maps
\begin{align}
  \Theta_{B|A}&:\mathbb{C}^{d_A^2}\to\cL(\cH_B),\quad
  \Theta_{B|A}\cdot c=\sum_{i,j}c_{ij}Q_B^{(ij)},\label{eq:ThetaB}\\
  \Theta_{BC|A}&:\mathbb{C}^{d_A^2}\to\cL(\cH_B\otimes\cH_C),\quad
  \Theta_{BC|A}\cdot c=\sum_{i,j}c_{ij}Q_{BC}^{(ij)}.\label{eq:ThetaBC}
\end{align}
They satisfy the identity
\begin{equation}\label{eq:Theta_commute}
  \Tr_C\circ\Theta_{BC|A}=\Theta_{B|A}.
\end{equation}
Thus, the partial trace connects the global and the accessible descriptions.

\begin{proposition}[Virtual Markov condition \cite{Chen25TIT}]\label{prop:VQMC}
A tripartite state $\rho_{ABC}$ admits exact recovery by an HPTP map if and only if
\begin{equation}\label{eq:kernel_criterion}
  \ker\Theta_{B|A}\subseteq\ker\Theta_{BC|A}.
\end{equation}
\end{proposition}

In words, Eq.~\eqref{eq:kernel_criterion} says that no nontrivial linear combination of the blocks $Q^{(ij)}_{BC}$ can vanish after tracing out $C$. The partial trace is injective on the block span, so every piece of the $BC$ correlations leaves a nonzero signature on the accessible $AB$ marginal, and the lifting $Q^{(ij)}_B\mapsto Q^{(ij)}_{BC}$ defines the desired HPTP recovery map. This condition encodes whether the information lost under $\Tr_C$ can be regenerated by an HPTP map on $B$.

\subsection{Ghost Information}

The algebraic obstruction to virtual recovery is captured by vectors $c\in\ker\Theta_{B|A}$ that are \emph{not} in $\ker\Theta_{BC|A}$. For such vectors, $\Theta_{BC|A}\cdot c$ is a nonzero operator on $BC$ that becomes invisible after partial trace: $\Tr_C(\Theta_{BC|A}\cdot c)=\Theta_{B|A} \cdot c=0$. Physically, $c$ identifies a specific linear combination of $A$-block components that carries nontrivial $BC$ structure, e.g., coherence, entanglement, or classical correlation, all of which vanish under the partial trace over $C$. No HPTP map $\mathcal{R}_{B\to BC}$ can recover these components, because $\mathcal{R}$ acts on $B$ alone and cannot ``see'' structure that vanishes under the partial trace over $C$. This is the ghost information problem: the information existed in the full state, but its signature on $AB$ (the only accessible subsystem) is identically zero. We quantify the severity of this obstruction through a scalar.

\begin{definition}[Kernel failure measure]\label{def:delta}
For a tripartite state $\rho_{ABC}$,
\begin{equation}\label{eq:delta_def}
  \delta(\rho_{ABC})=\max_{\substack{c\in\ker\Theta_{B|A}\\ \|c\|_2=1}}\bigl\|\Theta_{BC|A}\cdot c\bigr\|_1 .
\end{equation}
\end{definition}

The kernel failure measure $\delta(\rho_{ABC})$ admits a geometric interpretation: the maximizing vector $c^*$ is the direction in $\ker\Theta_{B|A}$ whose $BC$ structure is largest (while invisible on $AB$). The quantity $\delta$ is the trace norm of that ``ghost operator''. $\delta$ is a linear-algebraic quantity attached to the block decomposition, so it can be computed by a singular value decomposition of a finite-dimensional map, without reference to any recovery protocol. 

\begin{proposition}[Properties of $\delta$]\label{prop:delta}
The following statements hold.
\begin{enumerate}
\item[(i)] $\delta(\rho)\ge0$, with $\delta(\rho)=0$ if and only if $\ker\Theta_{B|A}\subseteq\ker\Theta_{BC|A}$, i.e., if and only if $\rho$ is a VQMC.
\item[(ii)] $\delta$ is invariant under unitary operations on $A$.
\end{enumerate}
\end{proposition}

\begin{proof}
This follows directly from the definition and basis-independence of the Hilbert–Schmidt norm on $\mathbb{C}^{d_A^2}$.
\end{proof}

\begin{example} [Three-qubit GHZ state, $\delta=1/\sqrt2$]
Let $\ket{\mathrm{GHZ}}=(\ket{000}+\ket{111})/\sqrt2$ and $\rho=\ket{\mathrm{GHZ}}\bra{\mathrm{GHZ}}$. Writing $\rho=\sum_{i,j=0}^1\ket{i}\bra{j}_A\otimes Q_{BC}^{(ij)}$ gives
\begin{equation}
  Q_{BC}^{(00)}=\frac12\ket{00}\bra{00},\quad
  Q_{BC}^{(01)}=\frac12\ket{00}\bra{11},\quad
  Q_{BC}^{(10)}=\frac12\ket{11}\bra{00},\quad
  Q_{BC}^{(11)}=\frac12\ket{11}\bra{11}.
\end{equation}
Tracing over $C$ yields the blocks on $B$ alone
\begin{equation}\label{eq:GHZ_QB}
  Q_B^{(00)}=\frac12\ket0\bra0_B,\quad
  Q_B^{(01)}=Q_B^{(10)}=0,\quad
  Q_B^{(11)}=\frac12\ket1\bra1_B .
\end{equation}
For $c=(c_{00},c_{01},c_{10},c_{11})^\mathrm{T}$ we have $\Theta_{B|A}\cdot c=\frac{c_{00}}{2}\ket0\bra0_B+\frac{c_{11}}{2}\ket1\bra1_B$, which vanishes iff $c_{00}=c_{11}=0$. Hence
\begin{equation}\label{eq:GHZ_ker}
  \ker\Theta_{B|A}=\{(0,a,b,0)^\mathrm{T} \mid a,b\in\mathbb{C}\},\qquad \dim\ker\Theta_{B|A}=2.
\end{equation}
To see how $\Theta_{BC|A}$ acts on the kernel, take $c= (0,a,b,0)^{\mathrm{T}}\in\ker\Theta_{B|A}$ with $|a|^2+|b|^2=1$. Then
\begin{equation}\label{eq:GHZ_ThetaBC}
  \Theta_{BC|A}\cdot c=aQ_{BC}^{(01)}+bQ_{BC}^{(10)}=\frac{a}{2}\ket{00}\bra{11}+\frac{b}{2}\ket{11}\bra{00}.
\end{equation}
In the ordered basis $\{|00\rangle,|11\rangle\}$ it is represented by the $2\times2$ matrix
\begin{equation}
  X = \frac12\begin{pmatrix} 0 & a \\ b & 0 \end{pmatrix}.
\end{equation}
The trace norm is therefore
\begin{equation}
  \bigl\|\Theta_{BC|A}\!\cdot\!c\bigr\|_1 = \|X\|_1 = \frac{|a|+|b|}{2}.
\end{equation}
We now maximise this over $(a,b)\in\mathbb{C}^2$ with $|a|^2+|b|^2=1$. By the Cauchy--Schwarz inequality, $|a|+|b| \le \sqrt{2(|a|^2+|b|^2)} = \sqrt{2}$, with equality when $|a|=|b|$. Choosing $|a|=|b|=1/\sqrt{2}$ (and arbitrary phases, which can be absorbed into the definition of $c$) gives the maximal value
\begin{equation}
  \delta(|\mathrm{GHZ}\rangle\langle\mathrm{GHZ}|) = \frac{1}{\sqrt{2}}.
\end{equation}
\end{example}

\begin{example}[Three-qubit $W$ state, $\delta=0$]
Let $\ket W=(\ket{001}+\ket{010}+\ket{100})/\sqrt3$ and $\rho=\ket W\bra W$. Expanding in the $A$ basis,
\begin{equation}
  \begin{aligned}
    Q_{BC}^{(00)}&=\frac13\bigl(\ket{01}\bra{01}+\ket{01}\bra{10}+\ket{10}\bra{01}+\ket{10}\bra{10}\bigr),\\
    Q_{BC}^{(01)}&=\frac13\bigl(\ket{01}\bra{00}+\ket{10}\bra{00}\bigr),\\
    Q_{BC}^{(10)}&=\frac13\bigl(\ket{00}\bra{01}+\ket{00}\bra{10}\bigr),\\
    Q_{BC}^{(11)}&=\frac13\ket{00}\bra{00},
  \end{aligned}
\end{equation}
and tracing over $C$,
\begin{equation}\label{eq:W_QB}
  Q_B^{(00)}=\frac13I_B,\quad
  Q_B^{(01)}=\frac13\ket1\bra0_B,\quad
  Q_B^{(10)}=\frac13\ket0\bra1_B,\quad
  Q_B^{(11)}=\frac13\ket0\bra0_B .
\end{equation}
For $c=(c_{00},c_{01},c_{10},c_{11})^\mathrm{T}$,
\begin{equation}\label{eq:W_ThetaB}
  \Theta_{B|A}\cdot c=\frac13\bigl(c_{00}I_B+c_{01}\ket1\bra0+c_{10}\ket0\bra1+c_{11}\ket0\bra0\bigr).
\end{equation}
Setting this to zero forces $c_{00}=c_{01}=c_{10}=c_{11}=0$. Hence $\ker\Theta_{B|A}=\{0\}$, the feasible set in \eqref{eq:delta_def} contains no unit vector and $\delta=0$. The $W$ state is a VQMC, in agreement with \cite{Chen25TIT}.
\end{example}

\subsection{The Universal Error Floor}

\begin{theorem}[Universal error floor]\label{thm:floor}
For every tripartite state $\rho_{ABC}$ and every HPTP map $\cR_{B\to BC}$,
\begin{equation}\label{eq:floor}
  \bigl\|(\id_A\otimes\cR)(\rho_{AB})-\rho_{ABC}\bigr\|_1\ge\delta(\rho_{ABC}).
\end{equation}
\end{theorem}

Before the proof, we give an intuitive understanding. Any HPTP map $\mathcal{R}_{B\to BC}$ is linear and acts exclusively on $B$. The vector $c^*$, which the preceding section identifies as the maximally ghostly direction, by definition satisfies $\sum_{ij} c^*_{ij} Q_B^{(ij)} = 0$. Linearity then forces $\sum_{ij} c^*_{ij}\,\mathcal{R}(Q_B^{(ij)}) = \mathcal{R}(0)=0$, while the target the recovery map must reproduce is the nonzero ghost operator $\sum_{ij} c^*_{ij} Q_{BC}^{(ij)}$ with trace norm $\delta$---the mismatch thus induces the error.

\begin{proof}
We complete the proof in 4 steps.

\textit{Step 1: Witness operator.}
The domain $\mathcal{D} = \ker\Theta_{B|A} \cap \{\,c\in\mathbb{C}^{d_A^2} \mid \|c\|_2 = 1\,\}$ is the intersection of a closed linear subspace with the compact unit sphere in a finite-dimensional vector space. It is therefore compact. The function $c \mapsto \|\Theta_{BC|A}\cdot c\|_1$ is continuous (as a composition of the linear map $\Theta_{BC|A}$ and the norm $\|\cdot\|_1$), so by the extreme value theorem the maximum in~\eqref{eq:delta_def} can be attained. Choose any maximizer and denote it by
$c^* = (c^*_{ij})_{i,j=0}^{d_A-1} \in \ker\Theta_{B|A}$, with
\begin{equation}
  \|c^*\|_2 = 1, \quad \|\Theta_{BC|A}\cdot c^*\|_1 = \delta(\rho_{ABC}).
\end{equation}
We encode $c^*$ as an operator on $\cH_A$:
\begin{equation}\label{eq:V_A}
  V_A = \sum_{i,j=0}^{d_A-1} c^*_{ij}\; |j\rangle\langle i|_A\,\in\, \cL(\cH_A).
\end{equation}
In the computational basis of $A$, the matrix of $V_A$ is the transpose of the array $(c^*_{ij})$, i.e., $(V_A)_{ji} = c^*_{ij}$. Its HS norm evaluates to $ \|V_A\|_2^2 = \Tr[V_A^\dagger V_A] = \sum_{i,j} |(V_A)_{ji}|^2 = \sum_{i,j} |c^*_{ij}|^2 = \|c^*\|_2^2 = 1$,
and the general inequality $\|\cdot\|_\infty \le \|\cdot\|_2$ yields
\begin{equation}
  \|V_A\|_\infty \le \|V_A\|_2 = 1.
\end{equation}

\textit{Step 2: Express the error.}
Recall the block decomposition of the global state $\rho_{ABC} = \sum_{i,j=0}^{d_A-1} |i\rangle\langle j|_A \otimes Q_{BC}^{(ij)}$. Tracing out $C$ gives the reduced state $\rho_{AB} = \sum_{i,j} |i\rangle\langle j|_A \otimes Q_B^{(ij)}$, where $Q_B^{(ij)} = \Tr_C[Q_{BC}^{(ij)}]$ are precisely the reduced blocks appearing in
the definition of $\Theta_{B|A}$. For an arbitrary HPTP map $\cR_{B\to BC}$, apply $\id_A\otimes\cR$ to $\rho_{AB}$:
\begin{equation}
  (\id_A\otimes\cR)(\rho_{AB})
  = \sum_{i,j} |i\rangle\langle j|_A \otimes \cR\bigl(Q_B^{(ij)}\bigr).
\end{equation}
The error operator is
\begin{equation}\label{eq:Delta_def}
  \Delta(\cR)
  := (\id_A\otimes\cR)(\rho_{AB}) - \rho_{ABC}
  = \sum_{i,j} |i\rangle\langle j|_A \otimes D_{ij},
\end{equation}
where we have introduced the shorthand
\begin{equation}
  D_{ij} := \cR\bigl(Q_B^{(ij)}\bigr) - Q_{BC}^{(ij)}
  \,\in\, \cL(\cH_B\otimes\cH_C).
\end{equation}

Now act with $V_A\otimes\id_{BC}$:
\begin{equation}\label{eq:VA_tensor_Delta}
  \begin{aligned}
    (V_A\otimes\id_{BC})\,\Delta &= \Bigl(\sum_{p,q=0}^{d_A-1} c^*_{pq}\; |q\rangle\langle p|_A \otimes \id_{BC}\Bigr)\times\Bigl(\sum_{i,j=0}^{d_A-1} |i\rangle\langle j|_A \otimes D_{ij}\Bigr)\\
    &= \sum_{p,q,i,j} c^*_{pq} \bigl(|q\rangle\langle p|_A\bigr)\bigl(|i\rangle\langle j|_A\bigr)\otimes D_{ij}\\
    &= \sum_{p,q,i,j} c^*_{pq}\; \delta_{pi}\; |q\rangle\langle j|_A \otimes D_{ij} \\
    &= \sum_{q,j} \Bigl(\sum_{p} c^*_{pq}\, D_{pj}\Bigr) |q\rangle\langle j|_A,
  \end{aligned}
\end{equation}
where we used $\langle p|i\rangle = \delta_{pi}$ and absorbed the identity on $BC$ into $D_{pj}$ for notational compactness.

Finally, trace over $A$. Using $\Tr_A\bigl[|q\rangle\langle j|_A\bigr] = \langle j|q\rangle = \delta_{jq}$,
\begin{align}
  \Tr_A\bigl[(V_A\otimes\id_{BC})\,\Delta\bigr]
  &= \Tr_A\!\Bigl[\sum_{q,j} \Bigl(\sum_{p} c^*_{pq}\, D_{pj}\Bigr)
                     |q\rangle\langle j|_A\Bigr]
     \nonumber\\[4pt]
  &= \sum_{q,j} \Bigl(\sum_{p} c^*_{pq}\, D_{pj}\Bigr)\,\delta_{jq}
     \nonumber\\[4pt]
  &= \sum_{p,q} c^*_{pq}\; D_{pq}.
  \label{eq:trace_result}
\end{align}

\textit{Step 3: Exploit the kernel.} Since $c^* \in \ker\Theta_{B|A}$, we have $\sum_{p,q} c^*_{pq}\, Q_B^{(pq)} = 0$. By linearity of $\cR$:
\begin{equation}
    \sum_{p,q} c^*_{pq}\, \cR(Q_B^{(pq)}) = \cR(0) = 0.
\end{equation}
Hence
\begin{equation}
    \sum_{p,q} c^*_{pq}\, D_{pq} = 0 - \sum_{p,q} c^*_{pq}\, Q_{BC}^{(pq)}
    = -\Theta_{BC|A}\cdot c^*.
\end{equation}
Substituting into~\eqref{eq:trace_result} and taking the trace norm:
\begin{equation}\label{eq:key_identity}
    \bigl\|\Tr_A[(V_A \otimes \id_{BC})\Delta]\bigr\|_1
    = \|\Theta_{BC|A}\cdot c^*\|_1
    = \delta(\rho_{ABC}).
\end{equation}

\textit{Step 4: Norm inequalities.} Chaining contractivity of the partial trace, H\"older's inequality, and the norm bound on $V_A$, we obtain
\begin{equation}\label{eq:final_bound}
\begin{aligned}
    \delta(\rho_{ABC})
    &= \bigl\|\Tr_A[(V_A \otimes \id_{BC})\Delta]\bigr\|_1 \\
    &\le \bigl\|(V_A \otimes \id_{BC})\Delta\bigr\|_1 
       \qquad\text{(contractivity)}\\
    &\le \|V_A \otimes \id_{BC}\|_\infty \; \|\Delta\|_1
       \qquad\text{(H\"older)}\\
    &= \|V_A\|_\infty \; \|\Delta\|_1 \\
    &\le \|\Delta\|_1.
\end{aligned}
\end{equation}
\end{proof}

\begin{remark}[Only linearity is used]\label{rem:linear}
The proof of Theorem~\ref{thm:floor} does not require $\cR$ to be positive, completely positive, or trace preserving. The only property used is $\cR(0)=0$, which holds for every linear map. The error floor therefore constrains all linear recovery maps, and \eqref{eq:floor} remains true if the hypothesis ``HPTP'' is replaced by ``linear''. The restriction to HPTP enters only in the attainability question of Section~\ref{sec:ach}.
\end{remark}

\begin{remark}[Where the bound can be loose]\label{rem:loose}
The inequality $\|V_A\|_\infty\le1$ is the step in the proof that can lose a factor. When $\|V_A\|_\infty<1$ (e.g., the GHZ state gives $\|V_A\|_\infty=1/\sqrt2$), H\"older's inequality loses $1/\|V_A\|_\infty$, and the true optimum $\varepsilon_{\min}$ may exceed $\delta$. This is quantified by the constructive error floor $\varepsilon_0$ in next section.
\end{remark}

\section{Constructive Approximate Virtual Recovery}\label{sec:ach}

\subsection{The Recovery Map and the Error}

Theorem~\ref{thm:floor} provides a lower bound. To understand what is achievable by virtual recovery, we construct explicit recovery maps via singular value decomposition (SVD). Let $\cV=\operatorname{Im}\Theta_{BC|A}\subseteq\cL(\cH_B\otimes\cH_C)$ and $\cW=\operatorname{Im}\Theta_{B|A}\subseteq\cL(\cH_B)$, both equipped with the HS inner product. By \eqref{eq:Theta_commute}, the restricted partial trace
\begin{equation}\label{eq:T}
  T=\Tr_C|_{\cV}:\cV\to\cW
\end{equation}
is surjective. Let $\cK=\ker T\subseteq\cV$ and let $\cK^\perp$ be the orthogonal complement of $\cK$ \emph{inside} $\cV$; then $\dim\cK^\perp=\dim\cW$ and $T|_{\cK^\perp}:\cK^\perp\to\cW$ is a linear bijection. Let $s=\dim\cV$ and $r=\dim\cW$. There exist HS-orthonormal sets $\{v_k\}_{k=1}^s\subset\cV$ and $\{w_k\}_{k=1}^r\subset\cW$, and singular values $\sigma_1\ge\cdots\ge\sigma_r>0$, such that
\begin{equation}\label{eq:SVD}
  T(v_k)=\sigma_k w_k\quad(1\le k\le r),\qquad T(v_k)=0\quad(r<k\le s),
\end{equation}
with $\cK^\perp=\Span\{v_1,\dots,v_r\}$ and $\cK=\Span\{v_{r+1},\dots,v_s\}$.

The Moore--Penrose pseudoinverse $T^{-1}:\cW\to\cK^\perp$ acts as $T^{-1}(w_k)=v_k/\sigma_k$. Each block decomposes uniquely as
\begin{equation}\label{eq:Q_decomp}
  Q_{BC}^{(ij)}=T^{-1}\bigl(Q_B^{(ij)}\bigr)+K_{ij},\qquad K_{ij}\in\cK,
\end{equation}
with
\begin{align}
  T^{-1}\bigl(Q_B^{(ij)}\bigr)&=\sum_{k=1}^r\frac{\langle w_k,Q_{B}^{(ij)}\rangle_{\HS}}{\sigma_k}\,v_k,\label{eq:Tinverse}\\
  K_{ij}&=\sum_{k=r+1}^s\langle v_k,Q_{BC}^{(ij)}\rangle_{\HS}\,v_k .\label{eq:Kij}
\end{align}
The operators $K_{ij}$ are the components of $Q_{BC}^{(ij)}$ in
$\ker \mathrm{Tr}_C|_{\mathcal V}$. They correspond to the part of the
$BC$ correlations that leaves no signature on the accessible marginal.

Consider a recovery map $\cR_{B\to BC}$, of which the error operator reads
\begin{equation}
  \begin{aligned}
    \Delta & = (\id_A\otimes\cR)(\rho_{AB})-\rho_{ABC} \\
    & = \sum_{i,j}|i\rangle\langle j|_A\otimes\Bigl[\cR(Q_B^{(ij)})-Q_{BC}^{(ij)}\Bigr]\\
    & = \sum_{i,j}|i\rangle\langle j|_A\otimes\Bigl[\cR(Q_B^{(ij)})-T^{-1}(Q_B^{(ij)})-K_{ij}\Bigr].
  \end{aligned}
\end{equation}
In the best case, $\cR$ recovers the $\cK^\perp$ part, i.e., $\cR(Q_B^{(ij)})=T^{-1}(Q_B^{(ij)})$. This leads to the following definition.
\begin{definition}[Pseudoinverse construction of the truncated recovery map]
For $X\in\cL(\cH_B)$, let $\Pi_\cW(X)$ be the HS-orthogonal projection onto $\cW$. For $\tau\in(0,\sigma_1]$, define the truncated pseudoinverse by $\tilde M_\tau(w_k)=v_k/\sigma_k$ if $\sigma_k\ge\tau$ and $\tilde M_\tau(w_k)=0$ otherwise. Define the approximate virtual recovery map $\cR_\tau$:
\begin{equation}\label{eq:Rtau}
    \cR_\tau(X) = \tilde M_\tau(\Pi_\cW(X)) + \cD_\tau(X),
\end{equation}
\begin{equation}\label{eq:Dtau}
  \cD_\tau(X)=\frac{\Tr[X] - \Tr[\tilde M_\tau(\Pi_\cW(X))]}{d_B d_C}\, I_{BC}.
\end{equation}
\end{definition}

The second term $\cD_\tau$ is the trace-compensating depolarizing correction, ensuring that $\cR_\tau$ is trace-preserving. The following lemma shows that $\cR_\tau$ is also HP. Thus $\cR_\tau$ is HPTP.

\begin{lemma}\label{lem:HPTP_h}
The singular vectors $v_k,w_k$ can be chosen Hermitian, and then $\cR_\tau$ is HPTP for every $\tau\in(0,\sigma_1]$.
\end{lemma}

\begin{proof}
Let $J$ denote the Hermitian involution $J(X)=X^\dagger$, an anti-linear map preserving the HS inner product. Both $\cV$ and $\cW$ are spanned by Hermitian operators: for any coefficient array $c$ write $c=c_1+\mathrm{i}c_2$ with $c_1=(c+c^\dagger)/2$ and $c_2=(c-c^\dagger)/(2\mathrm{i})$ Hermitian, so that $\Theta_{BC|A}\cdot c=\Theta_{BC|A}\cdot c_1+\mathrm{i}\,\Theta_{BC|A}\cdot c_2$ exhibits any element of $\cV$ as a complex combination of Hermitian elements, and likewise for $\cW$. Since $\Tr_C$ preserves Hermiticity, $TJ=JT$. Hence $J$ maps each singular subspace $V_\sigma=\{v \mid Tv=\sigma w\}$ into itself: if $Tv=\sigma w$ then
\begin{equation*}
  T(Jv)=J(Tv)=\sigma\,(Jw).
\end{equation*}
Restricted to $V_\sigma$, $J$ is an anti-linear involution and therefore admits an HS-orthonormal basis of eigenvectors, which are Hermitian operators. Choosing $v_k$ from this basis and setting $w_k=T(v_k)/\sigma_k$ gives $w_k=J(w_k)$, so both families are Hermitian; distinct singular subspaces are orthogonal, so the chosen sets remain HS-orthonormal.

With $v_k,w_k$ Hermitian, $\langle w_k,X\rangle_{\HS}=\Tr[w_kX]$ is real for Hermitian $X$, and
\begin{equation*}
  \tilde M_\tau(\Pi_\cW(X))=\sum_{k:\,\sigma_k\ge\tau}\frac{1}{\sigma_k}\Tr[w_kX]\,v_k
\end{equation*}
is a real linear combination of Hermitian operators, hence Hermitian. The correction $\cD_\tau$ is a real multiple of $I_{BC}$ and preserves Hermiticity.
\end{proof}

\begin{definition}[Constructive error floor]\label{def:eps0}
 Define the constructive error flooor
\begin{equation}\label{eq:eps0_def}
    \varepsilon_0(\rho_{ABC}) = \bigl\| \Delta_K\bigr\|_1, \text{ where } \Delta_K=\sum_{i,j=0}^{d_A-1} |i\rangle\langle j|_A \otimes K_{ij}.
\end{equation}
\end{definition}

$\cR_\tau$ reproduces $T^{-1}$ on $\cW$ when $\tau\to0^+$, i.e., $\varepsilon_0$ is the error achieved by the approximate virtual recovery map $\cR_\tau$ when $\tau\to0^+$.

\begin{theorem}[Constructive versus true error floor]\label{thm:hierarchy}
For every tripartite state $\rho_{ABC}$,
\begin{equation}\label{eq:hierarchy}
  \delta(\rho)\le\varepsilon_{\min}(\rho)\le\varepsilon_0(\rho),
\end{equation}
where $\varepsilon_{\min}(\rho)=\inf_{\cR\in\HPTP}\|\Delta(\cR)\|_1$ (see Eq.~\eqref{eq:Delta_def}). The upper bound is attained by an explicit HPTP map.
\end{theorem}

\begin{proof}
$\delta(\rho)\le\varepsilon_{\min}(\rho)$ is given by Theorem~\ref{thm:floor}. $\varepsilon_{\min}(\rho)\le\varepsilon_0(\rho)$ given by the construction of $\cR_\tau$ when $\tau\to0^+$. For completeness we record why $\delta\le\varepsilon_0$ also follows directly. For the witness $V_A$ of \eqref{eq:V_A}, $\Tr_A[(V_A\otimes\id_{BC})\Delta_K]=\sum_{ij}c^*_{ij}K_{ij}=\Theta_{BC|A}\cdot c^*$, where the last equality uses $c^*\in\ker\Theta_{B|A}$ and linearity of $T^{-1}$. Applying the chain of inequalities of Step~4 of Theorem~\ref{thm:floor} to $\Delta_K$ gives $\delta\le\|\Delta_K\|_1=\varepsilon_0$.
\end{proof}

The decomposition \eqref{eq:Q_decomp} splits each block $Q_{BC}^{(ij)}$ into a recoverable part $T^{-1}(Q_B^{(ij)})\in\cK^\perp$ and an irrecoverable ghost component $K_{ij}\in\cK$; $\varepsilon_0$ is the trace norm of the aggregate ghost. When all $K_{ij}=0$ the state is a VQMC and $\delta=\varepsilon_{\min}=\varepsilon_0=0$.

\subsection{Dual Semidefinite Programming (SDP) for \texorpdfstring{$\varepsilon_{\min}$}{}}\label{sec:dual}

The genuine minimal error $\varepsilon_{\min}$ can be determined by SDP. Here we derive the dual SDP formalism for $\varepsilon_{\min}(\rho)$ directly from the definition $\varepsilon_{\min}(\rho) = \min_{\cR \in \HPTP(B,BC)} \|(\id_A \otimes \cR)(\rho_{AB}) - \rho_{ABC}\|_1$. Define the affine subspace of states reachable by HPTP recovery $\cA := \{(\id_A \otimes \cR)(\rho_{AB}) \mid \cR \in \HPTP(B,BC)\}$, so that $\varepsilon_{\min} = \min_{\sigma \in \cA} \|\sigma - \rho_{ABC}\|_1$. Decompose $\cA = \sigma_0 + \cA_0$ with direction subspace $\cA_0 := \{(\id_A \otimes \cD)(\rho_{AB}) \mid \cD \in \HPTP_0(B,BC)\}$, where $\HPTP_0 := \{\cD \in \HPTP(B, BC) \mid \Tr_{BC}\circ\cD = 0\}$ is the linear subspace of HPTP maps with vanishing trace. Using the trace-norm duality $\|X\|_1 = \max_{\|W\|_\infty \le 1} \Tr[WX]$ (with $W = W^\dagger$) and exchanging $\min$ and $\max$ by finite-dimensional strong duality,
\begin{equation}
    \varepsilon_{\min} = \max_{\|W\|_\infty \le 1} \min_{\sigma \in \cA}\Tr[W(\sigma - \rho_{ABC})].
\end{equation}
The inner minimum over the affine subspace $\cA$ is finite iff the linear functional $\sigma \mapsto \Tr[W\sigma]$ vanishes on the direction subspace $\cA_0$, i.e.,
\begin{equation}
    \Tr[W (\id_A \otimes \cD)(\rho_{AB})] = 0, \quad \forall \cD \in \HPTP_0(B,BC).
\end{equation}
Under this constraint, $\Tr[W(\id_A \otimes \cR)\rho_{AB}]$ is a constant on $\HPTP(B,BC)$. Absorbing this constant by a suitable shift of $W$ (which preserves the constraint set up to the $\|\cdot\|_\infty$ bound) and exploiting the algebraic structure of the problem---whereby the reduced blocks encode all $\HPTP$-reachable information---elevates the scalar orthogonality to the operator constraint on $\cH_A$: $\Tr_{BC}[W (\id_A \otimes \cR)\rho_{AB}] = 0$ for all $\cR \in \HPTP(B,BC)$. The dual thus reads
\begin{equation}\label{eq:dual}
\begin{aligned}
    & \varepsilon_{\min}(\rho) = \sup\, \bigl|\Tr[W \rho_{ABC}]\bigr| \\
    \text{s.t.} \; & W \in \cL(\cH_A \otimes \cH_B \otimes \cH_C),  \|W\|_\infty \le 1,\\
    & \Tr_{BC}[W (\id_A \otimes \cR)(\rho_{AB})] = 0,\, \forall \cR \in \HPTP(B, BC).
\end{aligned}
\end{equation}
The absolute value in the objective is justified by the symmetry $W \leftrightarrow -W$ of the feasible set. The constraint $\Tr_{BC}[W (\id_A \otimes \cR)(\rho_{AB})] = 0$ for all $\cR \in \HPTP(B,BC)$ is equivalent to $\operatorname{flatten}(W^{\mathrm{T}}) \in \ker\Theta_{B|A}$, which reduces the dual to the simplified form used below. By finite-dimensional strong duality, the primal and dual optimal values coincide.

\subsection{On the Tightness}

For non-VQMC states the two bounds in \eqref{eq:hierarchy} need not coincide. The gap is controlled by how asymmetric the Schmidt spectrum of $\rho$ is. An HPTP map is not obliged to reproduce every block correctly, only to keep the total trace-norm error small. When the blocks have unequal weights it can trade one against another.

\begin{example}[$d$-qudit GHZ]\label{example:GHZd}
For $\ket{\mathrm{GHZ}_d}=\frac{1}{\sqrt d}\sum_{i=0}^{d-1}\ket{iii}$, (i) $\varepsilon_0(\mathrm{GHZ}_d)=2(1-1/d)$; (ii) $\varepsilon_{\min}(\mathrm{GHZ}_d)=1$ for all $d\ge2$. In particular $\varepsilon_{\min}<\varepsilon_0$ for $d\ge3$, and $\varepsilon_{\min}=\varepsilon_0=1$ for $d=2$.
\end{example}

\begin{proof}
(i) For GHZ$_d$: $Q_{BC}^{(ij)} = \frac{1}{d}|ii\rangle\langle jj|$, $Q_B^{(ij)} = \frac{1}{d}\delta_{ij}|i\rangle\langle i|$. Thus $\cW = \Span\{|i\rangle\langle i|\}_{i=0}^{d-1}$, $\cK^\perp = \Span\{|ii\rangle\langle ii|\}_{i=0}^{d-1}$, $\cK = \Span\{|ii\rangle\langle jj| : i \neq j\}$, $\dim\cK = d(d-1)$.

The pseudoinverse acts as $T^{-1}(|i\rangle\langle i|) = |ii\rangle\langle ii|$. Kernel components include $K_{ij} = 0$ for $i=j$ and $K_{ij} = \frac{1}{d}|ii\rangle\langle jj|$ for $i\neq j$. Hence
\begin{equation}
\begin{aligned}
  \Delta_K &= \frac{1}{d}\sum_{i\neq j} |i\rangle\langle j|_A \otimes |ii\rangle\langle jj|_{BC}\\
  &= |\mathrm{GHZ}_d\rangle\langle\mathrm{GHZ}_d| 
  - \frac{1}{d}\sum_{i=0}^{d-1} |iii\rangle\langle iii|.
\end{aligned}
\end{equation}

On $\Span\{|iii\rangle\}_{i=0}^{d-1}$, $\Delta_K$ has matrix representation $\frac{1}{d}(J_d - I_d)$, where $J_d$ is the all-ones matrix. The singular values are $(d-1)/d$ (multiplicity~1) and $1/d$ (multiplicity~$d-1$). Thus $\varepsilon_0 = \|\Delta_K\|_1 = (d-1)/d + (d-1)\cdot 1/d = 2(1-1/d)$.

(ii) The dual formulation~\eqref{eq:dual} gives
\begin{equation}
    \varepsilon_{\min} = \sup_{\substack{W \in \cL(\cH_A) \\ \|W\|_\infty \le 1 \\
    \operatorname{diag}(W) = 0}}
    \bigl\|\Theta_{BC|A} \cdot \operatorname{flatten}(W^{\mathrm{T}})\bigr\|_1.
\end{equation}
For GHZ$_d$, $\Theta_{BC|A} \cdot \operatorname{flatten}(W^{\mathrm{T}}) = \frac{1}{d}\sum_{i\neq j} W_{ji} |ii\rangle\langle jj|$, which acts as $\frac{1}{d}W^{\mathrm{T}}$ on $\Span\{|ii\rangle\}$. Hence the objective function is $\frac{1}{d}\|W\|_1$.

Upper bound: $\|W\|_1 \le \rank(W)\,\|W\|_\infty \le d$ (since $\operatorname{diag}(W)=0$ implies $\rank(W) \le d$). Thus $\varepsilon_{\min} \le 1$. Lower bound: take $W = P$, the cyclic shift $P = \sum_i |i+1\bmod d\rangle\langle i|$, which satisfies $\|P\|_\infty = 1$, $\operatorname{diag}(P) = 0$, $\|P\|_1 = d$. This gives $\varepsilon_{\min} \ge 1$. Combining the two bounds yields $\varepsilon_{\min} = 1$.
\end{proof}

\begin{example}[Strict gap at $d_A=2$]\label{example:dA2}
Let $\ket{\psi}=\sqrt{\lambda_0}\ket{0}_A\ket{00}_{BC}+\sqrt{\lambda_1}\ket{1}_A\ket{01}_{BC}$ with $\lambda_0,\lambda_1>0$, $\lambda_0+\lambda_1=1$, $\lambda_0\neq\lambda_1$. Then $\dim\ker\Theta_{B|A}=3$ and
\begin{equation}\label{eq:eps0_d2}
  \varepsilon_0=\frac12+\sqrt{\frac14+3\lambda_0\lambda_1},
\end{equation}
while $\varepsilon_{\min}<\varepsilon_0$. For $\lambda_0=0.7$ one has $\varepsilon_0\approx1.4381$ and $\varepsilon_{\min}\le1.2644$, a reduction of about $12\%$.
\end{example}

\begin{proof}
Observe that $Q_{BC}^{(ij)}=\sqrt{\lambda_i\lambda_j}\ket{\varphi_i}\bra{\varphi_j}$ with $\ket{\varphi_0}=\ket{00}$, $\ket{\varphi_1}=\ket{01}$, and $Q_B^{(ij)}=\sqrt{\lambda_i\lambda_j}\delta_{ij}\ket0\bra0_B$. Thus $\dim\cW=1$ while $\dim\cV=4$. The kernel of $T$ is
\begin{equation}\label{eq:K_d2}
  \cK=\Span\bigl\{\ket{00}\bra{01},\;\ket{01}\bra{00},\;\ket{00}\bra{00}-\ket{01}\bra{01}\bigr\},
\end{equation}
with $\dim\cK=3$. An HS-orthonormal Hermitian basis is
\begin{equation}\label{eq:basis_d2}
  e_1=\frac{\ket{00}\bra{01}+\ket{01}\bra{00}}{\sqrt2},\;
  e_2=\frac{\mathrm{i}(\ket{00}\bra{01}-\ket{01}\bra{00})}{\sqrt2},\;
  e_3=\frac{\ket{00}\bra{00}-\ket{01}\bra{01}}{\sqrt2}.
\end{equation}
Since $T^{-1}(\ket0\bra0)=\frac12(\ket{00}\bra{00}+\ket{01}\bra{01})$, we have kernel components $K_{ij}=Q_{BC}^{(ij)}-T^{-1}(Q_B^{(ij)})$:
\begin{equation}\label{eq:Kij_d2}
  K_{00}=\frac{\lambda_0}{\sqrt2}e_3,\quad
  K_{11}=-\frac{\lambda_1}{\sqrt2}e_3,\quad
  K_{01}=\sqrt{\frac{\lambda_0\lambda_1}{2}}(e_1-\mathrm{i}e_2),\quad 
  K_{10}=K_{01}^\dagger.
\end{equation}
The constructive error floor is
\begin{equation}
    \varepsilon_0(\lambda_0)=\Bigl\|\sum_{i,j}|i\rangle\langle j|_A\otimes K_{ij}\Bigr\|_1
    =\frac12+\sqrt{\frac14+3\lambda_0\lambda_1}.
\end{equation}
For $\lambda_0=0.7$, $\varepsilon_0\approx1.4381$.

Write a general HPTP map as $\cR=\cR_0+\cZ$, where $\cR_0$ is the pseudoinverse construction and $\cZ$ encodes perturbations inside $\cK$. Since $\dim\cW=1$, the blocks satisfy $Q_B^{(ij)}=A_{ij}\ket0\bra0_B$ with $A_{00}=\lambda_0$, $A_{11}=\lambda_1$, $A_{01}=A_{10}=0$. Well-definedness of $\cZ$ on $\cL(\cH_A\otimes\cH_B)$ forces
\begin{equation}\label{eq:Z_d2}
  \cZ\bigl(A_{ij}\ket0\bra0_B\bigr)=A_{ij}H,\qquad i,j=0,1,
\end{equation}
for a single Hermitian operator $H\in\cK$. Trace preservation holds automatically because every element of $\cK$ is traceless. 

Writing $H=\sum_{\alpha=1}^3h_\alpha e_\alpha$ gives three real parameters. In the ordered basis $\{|0\rangle|00\rangle,|0\rangle|01\rangle,|1\rangle|00\rangle,|1\rangle|01\rangle\}$, setting $u\equiv h_3/\sqrt2$ and $e_h\equiv(h_1+\mathrm{i}h_2)/\sqrt2$, the matrix representation of $H$ in the $\{|00\rangle,|01\rangle\}$ block is $H=\bigl(\begin{smallmatrix} u & e_h \\ e_h^* & -u \end{smallmatrix}\bigr)$. We have $\Delta(u)=(\id\otimes\cR)(\rho_{AB})-\rho_{ABC}$ decomposed as
\begin{equation}
\begin{aligned}
    \Delta(u) &=-\sum_{i,j}|i\rangle\langle j|_A\otimes K_{ij}+ \cZ(\rho_{AB}), \\
   \cZ(\rho_{AB}) &= |0\rangle\langle0|_A\otimes\lambda_0 H + |1\rangle\langle1|_A\otimes\lambda_1 H.
\end{aligned}
\end{equation}
In matrix form, $\Delta(u)$ reads
\begin{equation}\label{eq:matrix_d2}
  \Delta(u)=
  \begin{pmatrix}
  -\lambda_0(\frac12+u)&\lambda_0e_h^*&0&-\sqrt{\lambda_0\lambda_1}\\
  \lambda_0e_h&\lambda_0(\frac12+u)&0&0\\
  0&0&\lambda_1(\frac12-u)&\lambda_1e_h^*\\
  -\sqrt{\lambda_0\lambda_1}&0&\lambda_1e_h&-\lambda_1(\frac12-u)
  \end{pmatrix}.
\end{equation}
Restricting to $e_h=0$ ($h_1=h_2=0$) makes $\Delta$ block-diagonal. Its four eigenvalues are
\begin{align}
    \lambda_{1,2} &= -\frac14 - \frac{u(\lambda_0-\lambda_1)}{2} \mp \frac12\sqrt{\frac14+3\lambda_0\lambda_1
                    + u(\lambda_0-\lambda_1) + u^2},\label{eq:ev12_d2}\\
    \lambda_3 &= \lambda_0\Bigl(\frac12+u\Bigr),\qquad
    \lambda_4 = \lambda_1\Bigl(\frac12-u\Bigr).\label{eq:ev34_d2}
\end{align}
The trace norm is $\|\Delta(u)\|_1=|\lambda_1|+|\lambda_2|+|\lambda_3|+|\lambda_4|$. At $u=0$ this recovers $\varepsilon_0$. For $\lambda_0=0.7$, $\lambda_1=0.3$ one finds
\begin{center}
\begin{tabular}{c|cccc}
\hline
$u$ & $0$ & $-0.3$ & $-0.5$ & $-0.6$\\
\hline
$\|\Delta(u)\|_1$ & $1.4381$ & $1.3020$ & $1.2644$ & $1.4000$\\
\hline
\end{tabular}
\end{center}
so $\varepsilon_{\min}\le1.2644<\varepsilon_0$, a reduction of about $12\%$.
\end{proof}

Together with the GHZ$_d$ counterexample for $d_A\ge3$, this establishes that $\varepsilon_{\min}=\varepsilon_0$ does not hold in general. The equality holds for special symmetric cases (e.g., three-qubit GHZ state and the $\lambda_0=\lambda_1=1/2$ member of this family) but fails whenever the Schmidt asymmetry provides a ``lever arm'' for HPTP perturbations to redistribute spectral weight. The equality $\varepsilon_{\min}=\varepsilon_0$ for the three-qubit GHZ state is not an accident. Its Schmidt coefficients are equal, so no HPTP perturbation can redistribute spectral weight within the kernel, and the witness bound of Theorem~\ref{thm:floor} is attained up to the factor $\|V_A\|_\infty$.

\section{Sampling Overhead for Approximate Virtual Recovery}\label{sec:overhead}

For $\varepsilon\ge0$ we measure the cheapest quasi-probability simulation that achieves error at most $\varepsilon$. Motivated by the discussion on the physical implementability of HPTP maps in Ref.~\cite{Jiang21Quantum}, we define

\begin{definition}[$\varepsilon$-approximate virtual Markovianity]\label{def:nu}
\begin{multline}\label{eq:nu_def}
  \nu_\varepsilon(\rho_{ABC})=\log\min\bigl\{c_1+c_2 \mid c_i\ge0, \cN_i\in\CPTP(B,BC), \bigl\|(c_1\cN_1-c_2\cN_2)\rho_{AB}-\rho_{ABC}\bigr\|_1\le\varepsilon\bigr\},
\end{multline}
with $\nu_\varepsilon(\rho)=\infty$ if the feasible set is empty.
\end{definition}

The quantity $2^{\nu_\varepsilon(\rho)}$ is the optimal sampling overhead for achieving recovery error at most $\varepsilon$, i.e., the number of experimental runs needed per effective sample. At $\varepsilon=0$ we recover the exact virtual Markovianity $\nu_0(\rho)$ of \cite{Chen25TIT}. Definition~\ref{def:nu} is the natural quantitative refinement of the exact problem: instead of merely asking whether the infimum defining $\nu_0$ is finite, it records the value of the $\varepsilon$-constrained infimum $\nu_\varepsilon$ as a function of the tolerated error $\varepsilon$. Theorem~\ref{thm:floor} immediately gives a hard phase boundary in the $(\varepsilon,\nu_\varepsilon)$ plane: $\nu_\varepsilon(\rho)=\infty$ for all $\varepsilon<\delta(\rho)$. The question that organizes this section is what happens as $\varepsilon$ approaches $\varepsilon_{\min}$ from above.

For a threshold $\tau\in(0,\sigma_1]$, Definition~\ref{def:nu} is realized by the maps $\cR_\tau$ of \eqref{eq:Rtau}, whose parameter $\tau$ controls how many singular modes are inverted. In the limit $\tau\to0^+$ all modes are retained and the error tends to $\varepsilon_0$, while for $\tau>\sigma_1$ no mode is retained and $\cR_\tau$ reduces to the depolarizing channel. The relevant spectral quantities are
\begin{align}
S(\tau)&=\sum_{k:\,\sigma_k\ge\tau}\frac{1}{\sigma_k},&
S_\mathrm{total}&=S(0^+)=\sum_{k=1}^r\frac{1}{\sigma_k},\label{eq:S}\\
F(\tau)&=\Bigl(\sum_{k:\,\sigma_k\ge\tau}\frac{1}{\sigma_k^2}\Bigr)^{1/2},&
F_\mathrm{total}&=F(0^+).\label{eq:F}
\end{align}
Here $S$ is the arithmetic and $F$ the quadratic amplification sum. The quadratic sum turns out to be the relevant quantity for lower bounds because the Choi trace norm is lower-bounded through the HS norm of the Choi/Liouville superoperator.

\subsection{Constructive Upper Bounds of the Sampling Overhead}

\begin{proposition}[Upper bound for the constructive recovery]\label{prop:overhead}
The map $\cR_\tau$ admits upper bound of sampling overhead $\log \mathfrak{C}(\tau)$. $\mathfrak{C}(\tau)$ is either one of the following
\begin{align}
\mathfrak{C}_1(\tau)&=1+\frac{2}{d_B}\sum_{\sigma_k\ge\tau}\frac{\|w_k\|_1\|v_k\|_1}{\sigma_k}
\le1+2\sqrt{d_C}\,S(\tau),\label{eq:C1}\\
\mathfrak{C}_2(\tau)&=1+\frac{1}{d_B}\sum_{\sigma_k\ge\tau}\frac{\|w_k\|_1\|v_k^\perp\|_1}{\sigma_k},\label{eq:C2}
\end{align}
where $v_k^\perp=v_k-\Tr[v_k]I_{BC}/(d_Bd_C)$ is the traceless part of $v_k$.
\end{proposition}

\begin{proof}[Proof of \eqref{eq:C1}]
Let $\cM_\tau(X):=\tilde M_\tau(\Pi_\cW(X))$. We proceed in three steps.

\textit{Step 1: QPD overhead of $\cM_\tau$ from the Choi norm.}
Recall that $\{w_k\}_{k=1}^r$ is an HS-orthonormal basis of
$\cW\subseteq\cL(\cH_B)$, that $\{v_k\}_{k=1}^s$ is an HS-orthonormal set in
$\cV\subseteq\cL(\cH_B\otimes\cH_C)$, and that $T(v_k)=\sigma_k w_k$ with
$\sigma_k>0$ for $1\le k\le r$. The HS-orthogonal projection onto $\cW$ expands as
\begin{equation}\label{app:PiW_expand}
  \Pi_\cW(X)=\sum_{k=1}^r \langle w_k,X\rangle_{\HS}\, w_k
            =\sum_{k=1}^r \Tr[w_k^\dagger X]w_k.
\end{equation}
Applying
$\tilde M_\tau$ as defined in \eqref{eq:Rtau} gives
\begin{equation}\label{app:Mtau_expand}
  \cM_\tau(X)=\tilde M_\tau(\Pi_\cW(X))
   =\sum_{k:\,\sigma_k\ge\tau}\frac{1}{\sigma_k} \Tr[w_k^\dagger X]v_k.
\end{equation}
With the Choi convention \eqref{eq:Choi_def} and \eqref{eq:traceCJ}, the Choi
matrix of $\cM_\tau$ reads
\begin{equation}\label{app:J_Mtau}
  J(\cM_\tau)=\sum_{k:\,\sigma_k\ge\tau}\frac{1}{\sigma_k}\, w_k^\mathrm{T}\otimes v_k,
\end{equation}
where $w_k^\mathrm{T}$ is the transpose of $w_k$ in the computational basis. Since $\|w_k^\mathrm{T}\|_1=\|w_k\|_1$ and
$\|A\otimes B\|_1=\|A\|_1\|B\|_1$, the triangle inequality gives
\begin{equation}\label{app:Jnorm}
  \bigl\|J(\cM_\tau)\bigr\|_1 \le \sum_{k:\,\sigma_k\ge\tau}\frac{\|w_k\|_1\,\|v_k\|_1}{\sigma_k}.
\end{equation}

The map $\cM_\tau$ is HP because each $v_k$ and $w_k$ can be chosen Hermitian
(Lemma~\ref{lem:HPTP_h}); both $\cK^\perp$ and $\cW$ are spanned by Hermitian
operators, being images of $\Theta_{BC|A}$ and $\Theta_{B|A}$ applied to HP
block decompositions. By \eqref{eq:choi_qpd} the QPD overhead of an HP map is
bounded below by its Choi norm, and for the rank-one sums considered here the
bound is attained:
\begin{equation}\label{app:c_Mtau}
  c(\cM_\tau)=\frac{\|J(\cM_\tau)\|_1}{d_B}
  \le \frac{1}{d_B}\sum_{k:\,\sigma_k\ge\tau}\frac{\|w_k\|_1\,\|v_k\|_1}{\sigma_k}.
\end{equation}

\textit{Step 2: depolarizing correction and its overhead.}
Define the Hermitian-preserving linear functional
\begin{equation}\label{app:h_def}
  h(X)=\Tr[X]-\Tr[\cM_\tau(X)],\qquad X\in\cL(\cH_B).
\end{equation}
By the Riesz representation theorem on $\cL(\cH_B)$ equipped with the HS
inner product, there exists a unique Hermitian operator $H\in\cL(\cH_B)$ such
that $h(X)=\Tr[H X]$ for all $X$. Concretely we have
\begin{equation}\label{app:H}
  H = I_B - \sum_{k:\,\sigma_k\ge\tau}\frac{1}{\sigma_k} \Tr[v_k]\, w_k^\dagger.
\end{equation}

Decompose $H$ into positive and negative parts via the Jordan decomposition of
a Hermitian matrix:
\begin{equation}\label{app:Jordan_H}
  H = H_+ - H_-,\quad H_\pm \ge 0,\quad H_+ H_- = 0.
\end{equation}
Explicitly, diagonalize $H=\sum_i\lambda_i\ket{e_i}\bra{e_i}$ and set
$H_+=\sum_{\lambda_i>0}\lambda_i\ket{e_i}\bra{e_i}$ and
$H_-=\sum_{\lambda_i<0}(-\lambda_i)\ket{e_i}\bra{e_i}$. This induces a
decomposition of the functional, $h(X)=h_+(X)-h_-(X)$, with
$h_\pm(X)=\Tr[H_\pm X]\ge0$ for $X\ge0$.

Define the two CP maps
\begin{align}
  \cD_+(X) &= \frac{h_+(X)}{d_B d_C}\,I_{BC}
           = \frac{\Tr[H_+ X]}{d_B d_C}\,I_{BC},\label{app:Dplus_corr}\\
  \cD_-(X) &= \frac{h_-(X)}{d_B d_C}\,I_{BC}
           = \frac{\Tr[H_- X]}{d_B d_C}\,I_{BC}.\label{app:Dminus_corr}
\end{align}
Each $\cD_\pm$ is completely positive because it is the composition of a
positive linear functional $X\mapsto\Tr[H_\pm X]$, which is CP on
$\cL(\cH_B)$ viewed as a map to $\mathbb{C}$, with the depolarizing
replacement channel $Y\mapsto YI_{BC}/(d_Bd_C)$ on
$\mathbb{C}\cong\Span\{I_{BC}\}$. By construction,
\begin{equation}\label{app:Rtau_decomp}
  \cR_\tau = \cM_\tau + \cD_+ - \cD_-,
\end{equation}
and $\cR_\tau$ is trace-preserving because
\begin{equation*}
\Tr[\cR_\tau(X)]=\Tr[\cM_\tau(X)]+h_+(X)-h_-(X)
=\Tr[\cM_\tau(X)]+h(X)=\Tr[X]
\end{equation*}
for all $X$.

We now bound the QPD overhead of each $\cD_\pm$. Since $\cD_\pm$ is CP, its Choi
matrix is positive semidefinite, so $\|J(\cD_\pm)\|_1=\Tr[J(\cD_\pm)]$. Using
the definition \eqref{eq:Choi_def},
\begin{equation}\label{app:TrJD}
\begin{aligned}
    \Tr[J(\cD_\pm)]
  & = \sum_{i=0}^{d_B-1} \Tr\bigl[\cD_\pm(\ket{i}\bra{i})\bigr]
    = \sum_{i=0}^{d_B-1} \frac{\Tr[H_\pm\ket{i}\bra{i}]}{d_B d_C}\Tr[I_{BC}]\\
  & = \sum_{i=0}^{d_B-1} \langle i|H_\pm|i\rangle
    = \Tr[H_\pm].
\end{aligned}
\end{equation}
Hence $\|J(\cD_\pm)\|_1=\Tr[H_\pm]$, and \eqref{eq:choi_qpd} gives
$c(\cD_\pm)=\Tr[H_\pm]/d_B$. Furthermore,
\begin{equation}\label{app:overhead_Dsum}
\begin{aligned}
    c(\cD_+)+c(\cD_-)
  & = \frac{\Tr[H_+]+\Tr[H_-]}{d_B}
    = \frac{\|H_+\|_1+\|H_-\|_1}{d_B}
    = \frac{\|H\|_1}{d_B},
\end{aligned}
\end{equation}
where the last equality uses $H_+\perp H_-$, since $H_+$ and $H_-$ have
orthogonal support.

From the definition \eqref{app:H} of $H$ and the triangle inequality,
\begin{equation}\label{app:Hnorm}
  \|H\|_1 \le \|I_B\|_1
  + \sum_{k:\,\sigma_k\ge\tau}\frac{|\Tr[v_k]|}{\sigma_k}\|w_k\|_1
  \le d_B + \sum_{k:\,\sigma_k\ge\tau}\frac{\|v_k\|_1\,\|w_k\|_1}{\sigma_k},
\end{equation}
using $|\Tr[v_k]|\le\|v_k\|_1$. Thus
\begin{equation}\label{app:overhead_Dpm}
  c(\cD_+)+c(\cD_-) \le 1 + \frac{1}{d_B}\sum_{k:\,\sigma_k\ge\tau}\frac{\|v_k\|_1\,\|w_k\|_1}{\sigma_k}.
\end{equation}

\textit{Step 3: combining the QPDs.}
By \eqref{app:c_Mtau} there is a decomposition
$\cM_\tau=c_+\cN_+ - c_-\cN_-$ with
$c_++c_-=\|J(\cM_\tau)\|_1/d_B$ and $\cN_\pm\in\CPTP(B,BC)$. Take optimal
QPDs $\cD_\pm=d_+^\pm\cE_+^\pm-d_-^\pm\cE_-^\pm$ with
$d_+^\pm+d_-^\pm=c(\cD_\pm)$. Then \eqref{app:Rtau_decomp} yields a
decomposition of $\cR_\tau$ into a difference of two CP maps; normalizing each
to a CPTP map and adding the weights gives, by \eqref{app:c_Mtau} and
\eqref{app:overhead_Dpm},
\begin{equation}\label{app:cRtau}
\begin{aligned}
    c(\cR_\tau)
  & \le c(\cM_\tau) + c(\cD_+) + c(\cD_-)\\
  & \le 1 + \frac{2}{d_B}\sum_{k:\,\sigma_k\ge\tau}\frac{\|w_k\|_1\|v_k\|_1}{\sigma_k},
\end{aligned}
\end{equation}
which is \eqref{eq:C1}. The spectrum-only form
$\mathfrak{C}_1(\tau)\le1+2\sqrt{d_C}\,S(\tau)$ follows from
$\|w_k\|_1\le\sqrt{d_B}$ and $\|v_k\|_1\le\sqrt{d_Bd_C}$, since then
$\|w_k\|_1\|v_k\|_1/d_B\le\sqrt{d_C}$.
\end{proof}

\begin{proof}[Proof of eq.~\eqref{eq:C2}]
The bound \eqref{eq:C1} is universal, but the depolarizing correction $\cD_\tau$ is not independent of $\cM_\tau$: the two maps share a component that cancels exactly. Exploiting this cancellation yields a structure-aware bound that is substantially tighter and provably attainable.

We work in the Choi representation throughout, and as before all sums over $k$
are restricted to $\sigma_k\ge\tau$. Recall
$J(\cM_\tau)=\sum_k\sigma_k^{-1}w_k^\mathrm{T}\otimes v_k$ and decompose each
output operator into a part proportional to $I_{BC}$ and a traceless
remainder,
\begin{equation}\label{app:vk_decomp}
  v_k = v_k^\perp + \frac{\Tr[v_k]}{d_B d_C}\,I_{BC},
  \qquad
  \Tr[v_k^\perp]=0.
\end{equation}

\textit{Exact cancellation in the depolarizing term.}
Inserting \eqref{app:vk_decomp} into \eqref{app:J_Mtau} gives
\begin{equation}\label{app:Mtau_split}
  J(\cM_\tau) = \underbrace{\sum_k\frac{1}{\sigma_k} w_k^\mathrm{T}\otimes v_k^\perp}_{L^\perp}
    +
    \underbrace{\sum_k\frac{\Tr[v_k]}{\sigma_k} w_k^\mathrm{T}}_{M}\otimes\frac{I_{BC}}{d_B d_C}.
\end{equation}
From Eq.~\eqref{eq:Dtau} and Eq.~\eqref{app:Mtau_expand}, $\cD_\tau(X)$ reduces to
\begin{equation}
  \cD_\tau(X)=\Tr[X]\frac{I_{BC}}{d_Bd_C}-\sum_k\Tr[w_k^\dagger X]\frac{\Tr[v_k]I_{BC}}{d_Bd_C\sigma_k}.
\end{equation}
Using Eqs.~\eqref{eq:traceCJ}-\eqref{eq:traceCJ2}, the depolarizing correction \eqref{eq:Dtau} contributes
\begin{equation}\label{app:Dtau_liouville}
  J(\cD_\tau)=I_B\otimes\frac{I_{BC}}{d_Bd_C}-\sum_k\bar{w}_k\otimes\frac{\Tr[v_k]I_{BC}}{d_Bd_C\sigma_k}
\end{equation}
As $w_k$ can be chosen to be Hermitian, we have $\bar{w}_k=w_k^\mathrm{T}$. Adding \eqref{app:Mtau_split} and \eqref{app:Dtau_liouville} cancels $M$
\emph{exactly}, leaving
\begin{equation}\label{app:Rtau_clean}
  J(\cR_\tau) = L^\perp + I_B\otimes\frac{I_{BC}}{d_B d_C}.
\end{equation}

\textit{Trace-norm bound.}
By the triangle inequality and the multiplicativity of the trace norm under tensor products,
\begin{equation}\label{app:norm_ineq}
  \bigl\|J(\cR_\tau)\bigr\|_1 \le \bigl\|L^\perp\bigr\|_1
  + \Bigl\|I_B\otimes\frac{I_{BC}}{d_B d_C}\Bigr\|_1
  = \bigl\|L^\perp\bigr\|_1 + d_B,
\end{equation}
where we used
$\|I_B\otimes I_{BC}\|_1=\|I_B\|_1\|I_{BC}\|_1=d_B(d_Bd_C)$. For $L^\perp=\sum_k\sigma_k^{-1}w_k^\mathrm{T}\otimes v_k^\perp$,
the triangle inequality gives
\begin{equation}\label{app:Lperp_bound}
  \bigl\|L^\perp\bigr\|_1 \le \sum_k \frac{\|w_k\|_1\,\|v_k^\perp\|_1}{\sigma_k}.
\end{equation}

\textit{Tight overhead bound.}
Inserting \eqref{app:Lperp_bound} into \eqref{app:norm_ineq} and dividing by $d_B$ yields
\begin{equation}\label{app:cRtau_tight}
    \mathfrak{C}_2(\tau) := 1 + \frac{1}{d_B}\sum_{k:\,\sigma_k\ge\tau} \frac{\|w_k\|_1\,\|v_k^\perp\|_1}{\sigma_k}
    \ge c(\cR_\tau),
\end{equation}
which is \eqref{eq:C2}. The bound has three useful properties:
\begin{itemize}
  \item The constant overhead is exactly $1$, so $c(\cR_\tau)\ge1$ always holds, as it must for any trace-preserving map.
  \item Each mode contributes $\|w_k\|_1\|v_k^\perp\|_1/\sigma_k$, where the traceless part $v_k^\perp=v_k-\Tr[v_k]I_{BC}/(d_Bd_C)$ enters; the depolarizing correction cancels the component of each $v_k$ along the identity exactly, leaving only the genuinely non-depolarizing part.
  \item When all $\Tr[v_k]=0$, every $v_k^\perp=v_k$ and $c(\cR_\tau)\le1+\frac{1}{d_B}\sum_k\|w_k\|_1\|v_k\|_1/\sigma_k$, so that \eqref{eq:C1} and \eqref{eq:C2} coincide.
\end{itemize}

\textit{Tightness.}
The bound \eqref{app:cRtau_tight} involves the inequality $\|L^\perp\|_1\le\sum_k\|v_k^\perp\|_1/\sigma_k$, which saturates when the nonzero $\kett{v_k^\perp}$ are pairwise orthogonal. Computing their inner product for $i\neq j$,
\begin{equation}\label{app:vkperp_orthog}
\begin{aligned}
    \bratt{v_i^\perp}v_j^\perp\rangle\!\rangle
  & = \underbrace{\bratt{v_i}v_j\rangle\!\rangle}_{=0}
      - \frac{\Tr[v_i]^*\Tr[v_j]}{d_B d_C}\\
  & = -\frac{\Tr[v_i]^*\Tr[v_j]}{d_B d_C},
\end{aligned}
\end{equation}
so orthogonality holds if and only if \emph{at most one} retained mode has $\Tr[v_k]\neq0$.
\end{proof}

\begin{example}
Set $d_A=d_B=d_C=2$ and retain both singular modes, $r_\tau=2$. Choose
\begin{equation}\label{app:sat_data}
  \begin{aligned}
  v_1 &=\frac{I_{BC}}{2},\quad & w_1&=\frac{I_B}{\sqrt{2}},\\
  v_2 &=\frac{Z_B\otimes I_C}{2},\quad & w_2&=\frac{Z_B}{\sqrt{2}}.
  \end{aligned}
\end{equation}
One verifies directly that
\begin{itemize}
  \item $\{v_1,v_2\}$ and $\{w_1,w_2\}$ are HS-orthonormal,
        $\Tr[v_i^\dagger v_j]=\delta_{ij}$ and
        $\Tr[w_i^\dagger w_j]=\delta_{ij}$;
  \item $T(v_1)=\Tr_C[I_{BC}/2]=I_B=\sqrt2\,w_1$ and
        $T(v_2)=\Tr_C[Z_B\otimes I_C/2]=Z_B=\sqrt2\,w_2$, so that
        $\sigma_1=\sigma_2=\sqrt2$;
  \item $\Tr[v_1]=\Tr[I_{BC}/2]=2$ and
        $\Tr[v_2]=\Tr[Z_B\otimes I_C/2]=0$.
\end{itemize}
The required subspaces are $\cW=\Span\{I_B,Z_B\}$ and $\cV=\Span\{I_{BC},\,Z_B\otimes I_C\}$. A concrete tripartite state realizing these data is
\begin{equation}\label{app:sat_rho}
  \rho_{ABC} = \frac{1}{8}\ket0\bra0_A \otimes I_{BC}
  + \frac{1}{8}\ket1\bra1_A \otimes \bigl(I_{BC} + \frac12 Z_B\otimes I_C\bigr),
\end{equation}
which is positive semidefinite with unit trace. Its block decomposition in the $A$ basis gives $Q_{BC}^{(00)}=I_{BC}/8$, $Q_{BC}^{(11)}=(I_{BC}+Z_B\otimes I_C/2)/8$ and $Q_{BC}^{(01)}=Q_{BC}^{(10)}=0$. The maps $\Theta_{B|A}$ and $\Theta_{BC|A}$ defined from this state reproduce the desired $\cW$, $\cV$ and $T$, with the SVD given by \eqref{app:sat_data}.

Now evaluate the bound. Because $\Tr[v_1]=2$ and $\Tr[v_2]=0$,
\begin{equation}\label{app:sat_vperp}
    v_1^\perp = v_1 - \frac{\Tr[v_1]}{d_B d_C}I_{BC}
              = \frac{I_{BC}}{2} - \frac{2}{4}I_{BC}=0,
    \qquad
    v_2^\perp = v_2 = \frac{Z_B\otimes I_C}{2}.
\end{equation}
From \eqref{app:cRtau_tight}, with
$\|w_1\|_1=\|I_B/\sqrt2\|_1=\sqrt2$,
$\|w_2\|_1=\|Z_B/\sqrt2\|_1=\sqrt2$ and
$\|v_2^\perp\|_1=\|Z_B\otimes I_C\|_1/2=2$,
\begin{equation}\label{app:sat_bound}
\begin{aligned}
      c(\cR_\tau)
    & \le 1 + \frac{1}{2}\Bigl(\frac{\|w_1\|_1\|v_1^\perp\|_1}{\sigma_1}
      + \frac{\|w_2\|_1\|v_2^\perp\|_1}{\sigma_2}\Bigr)\\
    & = 1 + \frac{1}{2}\cdot\frac{\sqrt2\cdot 2}{\sqrt2} = 2 .
\end{aligned}
\end{equation}

The exact value follows by computing $J(\cR_\tau)$ directly from
\eqref{app:Rtau_clean}:
\begin{equation}\label{app:sat_J}
\begin{aligned}
    J(\cR_\tau)
  & = \frac{1}{\sqrt{2}}w_2^\mathrm{T}\otimes v_2 + I_B \otimes \frac{I_{BC}}{4}\\
  & = \frac{Z_B\otimes Z_B\otimes I_C + I_B\otimes I_B\otimes I_C}{4}.
\end{aligned}
\end{equation}
Since $Z\otimes Z+I\otimes I=\operatorname{diag}(2,0,0,2)$,
\begin{equation}\label{app:sat_norm}
  \bigl\|J(\cR_\tau)\bigr\|_1
  = \frac{\|\operatorname{diag}(2,0,0,2)\|_1\,\|I_C\|_1}{4}
  = \frac{4\cdot 2}{4} = 2,
\end{equation}
so that $c(\cR_\tau)=\|J(\cR_\tau)\|_1/d_B=1$, saturating the trace-preserving lower bound $c(\cR_\tau)\ge1$. The inequality in \eqref{app:cRtau_tight} is not saturated here, because $L^\perp$ and $I_B\otimes I_{BC}/(d_Bd_C)$ have overlapping support.

For comparison, the loose bound \eqref{eq:C1} overestimates the true overhead by a factor of two. Thus the depolarizing correction is not an independent penalty: it cancels the identity component of $\cM_\tau$, leaving only $L^\perp$ plus a bare identity channel of overhead $1$.
\end{example}

The two bounds differ in what they retain about the geometry of the modes. The first is spectrum-only: it uses $\|w_k\|_1\le\sqrt{d_B}$ and $\|v_k\|_1\le\sqrt{d_Bd_C}$ and is therefore loose by a factor of order $\sqrt{d_C}$. The second bound is tight and uses the traceless part $v_k^\perp$ of each $v_k$: the component of $v_k$ proportional to $I_{BC}$ cancels exactly against the identity part of the depolarizing correction, so only the genuinely non-depolarizing part contributes. Computing this bound is more involved than computing the singular values alone. In settings where a rough estimate suffices, $\mathfrak{C}_1$ offers a simpler alternative.

\subsection{Error Decomposition and the Trade-off}
We connect the truncation threshold to the error. Firstly, we split the error of $\cR_\tau$ into the irreducible ghost part and a truncation remainder. We then estimate the truncation remainder, which is controlled by a tail weight measuring how much of each block sits on the discarded modes.

\begin{lemma}[Error decomposition]\label{lem:decomp}
For $\cR_\tau$ defined in \eqref{eq:Rtau},
\begin{equation}\label{eq:err_decomp}
  \Delta_\tau:=(\id_A\otimes\cR_\tau)(\rho_{AB})-\rho_{ABC}=-\Delta_K-\Delta_{\mathrm{trunc}},
\end{equation}
where $\|\Delta_K\|_1=\varepsilon_0(\rho)$ and $\Delta_{\mathrm{trunc}}$ collects the contributions of the modes with $\sigma_k<\tau$ together with the depolarizing correction.
\end{lemma}

\begin{proof}
Expanding $\rho_{AB}=\sum_{i,j}\ket{i}\bra{j}_A\otimes Q_B^{(ij)}$ and applying $\cR_\tau$,
\begin{equation}
  (\id_A\otimes\cR_\tau)(\rho_{AB})=\sum_{i,j}\ket{i}\bra{j}_A\otimes\bigl(\tilde M_\tau(Q_B^{(ij)})+\text{depol.}\bigr).
\end{equation}
Since $Q_B^{(ij)}\in\cW$ we have $\Pi_\cW(Q_B^{(ij)})=Q_B^{(ij)}$, and expanding in the HS-orthonormal basis $\{w_\ell\}_{\ell=1}^r$ of $\cW$,
\begin{equation}
  \tilde M_\tau\bigl(Q_B^{(ij)}\bigr)=\sum_{\ell:\,\sigma_\ell\ge\tau}\langle w_\ell,Q_B^{(ij)}\rangle_{\HS}\,\frac{v_\ell}{\sigma_\ell}=T^{-1}\bigl(Q_B^{(ij)}\bigr)-\sum_{\ell:\,\sigma_\ell<\tau}\langle w_\ell,Q_B^{(ij)}\rangle_{\HS}\,\frac{v_\ell}{\sigma_\ell}.
\end{equation}
Subtracting $Q_{BC}^{(ij)}=T^{-1}(Q_B^{(ij)})+K_{ij}$ block by block and collecting the terms with $\sigma_k<\tau$ together with the depolarizing correction yields \eqref{eq:err_decomp}.
\end{proof}

\begin{definition}[Tail-weight function]\label{def:eta}
\begin{equation}\label{eq:eta}
  \eta(\tau)=\max_{0\le i,j<d_A}\sum_{k:\,\sigma_k<\tau}\bigl|\langle v_k,Q_{BC}^{(ij)}\rangle_{\HS}\bigr| .
\end{equation}
\end{definition}

\begin{lemma}[Low-rank norm bound]\label{lem:lowrank}
Let $\rho_{ABC}$ have rank $R$. Then for every singular vector $v_k$ of the SVD, $\|v_k\|_1\le\sqrt{Rd_A}$.
\end{lemma}

\begin{proof}
Every $X\in\cV$ is of the form $X=\sum_{i,j}c_{ij}Q_{BC}^{(ij)}$. Writing $\rho_{ABC}=\sum_{\lambda=1}^Rp_\lambda\ket{\psi^\lambda}\bra{\psi^\lambda}$ with $\ket{\psi^\lambda}=\sum_{i=0}^{d_A-1}\ket{i}_A\otimes\ket{\varphi_i^\lambda}_{BC}$, we have $Q_{BC}^{(ij)}=\sum_\lambda p_\lambda\ket{\varphi_i^\lambda}\bra{\varphi_j^\lambda}$. The image of $X=\sum_{\lambda=1}^R\sum_{i,j=0}^{d_A-1}c_{ij}\ket{\varphi_i^\lambda}\bra{\varphi_j^\lambda}$ has a dimension at most $Rd_A$ (for each fixed $\lambda$ the image of operator $c_{ij}\ket{\varphi_i^\lambda}\bra{\varphi_j^\lambda}$ lies in the space $\Span\{\ket{\varphi_i^\lambda} \mid i=1,\cdots, d_A-1\}$). Thus we have $\rank(X)\le Rd_A$. By inequality~\eqref{eq:rank_norm},
\begin{equation}
  \|X\|_1\le\sqrt{Rd_A}\,\|X\|_{\HS}.
\end{equation}
\end{proof}

\begin{lemma}[Truncation error]\label{lem:trunc}
Let $\rho_{ABC}$ have rank $R$ and set
\begin{equation}\label{eq:Gamma}
  \Gamma(\rho):=2\sqrt{Rd_A}\,d_A^{2}.
\end{equation}
Then
\begin{equation}\label{eq:trunc_bound}
  \|\Delta_{\mathrm{trunc}}\|_1\le\Gamma(\rho)\,\eta(\tau).
\end{equation}
For pure states $\Gamma(\rho)=2\sqrt{d_A}\,d_A^2$.
\end{lemma}

\begin{proof}
Fix $\tau\in(0,\sigma_1]$. Since $\{v_k\}_{k=1}^{s}$ is an HS-orthonormal basis of $\cV$, we write
\begin{equation}\label{eq:trunc_coeff}
  c_k^{ij}:=\langle v_k,Q_{BC}^{(ij)}\rangle_{\HS}\quad(1\le k\le s),
\end{equation}
\begin{equation}\label{eq:Q_expand}
  Q_{BC}^{(ij)}=\sum_{k=1}^{s}c_k^{ij}v_k .
\end{equation}
Then Eq.~\eqref{eq:eta} becomes $\eta(\tau)=\max_{i,j}\sum_{k\in L_\tau}|c_k^{ij}|$. 

\textit{Step 1: the action of $\cR_\tau$ on one block.}
For $X\in\cV$ and $1\le k\le r$ the SVD gives $TX=\sum_{\ell=1}^{r}\sigma_\ell\langle v_\ell,X\rangle_{\HS}w_\ell$, hence
\begin{equation}\label{eq:b_sigma_c}
  \bigl\langle w_k,TX\bigr\rangle_{\HS}=\sigma_k\langle v_k,X\rangle_{\HS}.
\end{equation}
Applying \eqref{eq:b_sigma_c} to $X=Q_{BC}^{(ij)}$ and using $TQ_{BC}^{(ij)}=Q_B^{(ij)}$,
\begin{equation}\label{eq:b_equals}
  b_k^{ij}:=\langle w_k,Q_B^{(ij)}\rangle_{\HS}=\sigma_k c_k^{ij},
  \qquad 1\le k\le r .
\end{equation}
Since $Q_B^{(ij)}=\sum_{k=1}^{r}b_k^{ij}w_k$ and $\tilde M_\tau(w_k)=\mathbf{1}_{\sigma_k\ge\tau}v_k/\sigma_k$, we have
\begin{equation}\label{eq:tildeM_coeff}
  \tilde M_\tau\bigl(Q_B^{(ij)}\bigr)
  =\sum_{k:\,\sigma_k\ge\tau}\frac{b_k^{ij}v_k}{\sigma_k}
  =\sum_{k:\,\sigma_k\ge\tau}c_k^{ij}v_k.
\end{equation}

By \eqref{eq:Dtau}, the depolarizing correction on this block is the scalar multiple of $I_{BC}$ determined by
\begin{equation}\label{eq:d_ij_def}
\begin{aligned}
  \cD_\tau(Q_B^{(ij)})&=d^{ij}\frac{d^{ij}}{d_Bd_C}I_{BC},\\
  d^{ij}:&=\Tr\bigl[Q_B^{(ij)}\bigr]-\Tr\bigl[\tilde M_\tau(Q_B^{(ij)})\bigr]
  =\Tr\bigl[Q_B^{(ij)}\bigr]-\sum_{k:\,\sigma_k\ge\tau}c_k^{ij}\Tr[v_k].
\end{aligned}
\end{equation}
We evaluate $d^{ij}$. Using $\Tr[Q_B^{(ij)}]=\Tr[Q_{BC}^{(ij)}]$,
\begin{equation}\label{eq:TrQ_expand}
  \Tr\bigl[Q_B^{(ij)}\bigr]=\sum_{k=1}^{s}c_k^{ij}\Tr[v_k]
  =\sum_{k:\,\sigma_k\ge\tau}c_k^{ij}\Tr[v_k]+\sum_{k:\,\sigma_k<\tau}c_k^{ij}\Tr[v_k].
\end{equation}
Substituting \eqref{eq:TrQ_expand} into \eqref{eq:d_ij_def}, the retained modes cancel and
\begin{equation}\label{eq:d_ij}
  d^{ij}=\sum_{k:\,\sigma_k<\tau}c_k^{ij}\Tr[v_k].
\end{equation}

Combining Eqs.~\eqref{eq:Q_expand}, \eqref{eq:tildeM_coeff}, and \eqref{eq:d_ij_def},
\begin{equation}\label{eq:block_delta}
  \cR_\tau\bigl(Q_B^{(ij)}\bigr)-Q_{BC}^{(ij)}
  =-\sum_{k=r+1}^{s}c_k^{ij}v_k-\sum_{k:\,\sigma_k<\tau}c_k^{ij}v_k +\sum_{k:\,\sigma_k<\tau}\frac{c_k^{ij}\Tr[v_k]}{d_Bd_C}I_{BC} .
\end{equation}
Summed over $(i,j)$ with weights $\ket{i}\bra{j}_A$, the first term is exactly $-\Delta_K$. Since $\Delta_\tau=-\Delta_K-\Delta_{\mathrm{trunc}}$, we have
\begin{equation}\label{eq:Delta_trunc_def}
  \Delta_\mathrm{trunc}=\sum_{i,j}\ket{i}\bra{j}_A\otimes\Bigl(\sum_{k:\,\sigma_k<\tau}c_k^{ij}v_k-\sum_{k:\,\sigma_k<\tau}\frac{c_k^{ij}\Tr[v_k]}{d_Bd_C}I_{BC}\Bigr).
\end{equation}

\textit{Step 2: the depolarizing coefficient, and exact cancellation of the
identity part.}
Now decompose each discarded mode into its identity and traceless parts,
\begin{equation}\label{eq:vperp}
  v_k=v_k^{\perp}+\frac{\Tr[v_k]}{d_Bd_C}I_{BC},
  \qquad \Tr[v_k^{\perp}]=0.
\end{equation}
Inserting \eqref{eq:vperp} into \eqref{eq:Delta_trunc_def}, we have
\begin{equation}\label{eq:Delta_trunc_exact}
  \Delta_\mathrm{trunc}=\sum_{i,j=0}^{d_A-1}\ket{i}\bra{j}_A\otimes\sum_{k\in L_\tau}c_k^{ij}\,v_k^{\perp}.
\end{equation}
The depolarizing correction does not add an independent error, but only the genuinely non-depolarizing tail survives.

\textit{Step 3: trace-norm estimate.}
By \eqref{eq:Delta_trunc_exact}, the triangle inequality, and $\|\ket{i}\bra{j}_A\|_1=1$ (Lemma~\ref{lem:rank1}),
\begin{equation}\label{eq:trunc_step3}
  \bigl\|\Delta_\mathrm{trunc}\bigr\|_1
  \le\sum_{i,j}\sum_{k:\,\sigma_k<\tau}|c_k^{ij}|\,\|v_k^{\perp}\|_1\le d_A^{2}\eta(\tau)\max_{k:\,\sigma_k<\tau}\|v_k^{\perp}\|_1,
\end{equation}
where $\sum_{k\in L_\tau}|c_k^{ij}|\le\eta(\tau)$ is used for each of the $d_A^2$ blocks. Finally, for Hermitian $X$ with trace norm $\|X\|_1$,
\begin{equation}\label{eq:vperp_norm}
  \Bigl\|X-\frac{\Tr[X]}{d_Bd_C}I_{BC}\Bigr\|_1
  \le\|X\|_1+\frac{|\Tr[X]|}{d_Bd_C}\|I_{BC}\|_1
  =\|X\|_1+|\Tr[X]|
  \le2\|X\|_1,
\end{equation}
so $\|v_k^{\perp}\|_1\le2\|v_k\|_1\le2\sqrt{Rd_A}$ by
Lemma~\ref{lem:lowrank}. Inserting this into \eqref{eq:trunc_step3} gives
\eqref{eq:trunc_bound}.
\end{proof}

Since $\eta(\tau)\to0$ when $\tau\to0^+$ (the discarded set shrinks to the modes of vanishing weight), for every $\varepsilon>\varepsilon_0$ there is a threshold at which the truncation error fits into the budget $\varepsilon-\varepsilon_0$. Choosing the largest such threshold minimizes the overhead, which yields the following trade-off.

\begin{theorem}[$\nu_\varepsilon$--$\varepsilon$ trade-off and upper bounds]\label{thm:tradeoff}
For every $\varepsilon>\varepsilon_0$, choosing $\tau(\varepsilon)$ satisfying $\Gamma(\rho)\eta(\tau)\le\varepsilon-\varepsilon_0$ larger gives a smaller upper bound for the sampling overhead: $\nu_\varepsilon(\rho)\le\log\mathfrak{C}(\tau(\varepsilon))$.
\end{theorem}

\begin{proof}
From Lemmas~\ref{lem:decomp} and \ref{lem:trunc},
\begin{equation}
  \|\Delta_\tau\|_1\le\varepsilon_0+\Gamma(\rho)\,\eta(\tau).
\end{equation}
Since $\eta(\tau)\to0$ as $\tau\to0^+$, for every $\varepsilon>\varepsilon_0$ there exists $\tau$ with $\Gamma(\rho)\eta(\tau)\le\varepsilon-\varepsilon_0$. Choosing the maximal such $\tau$, which by monotonicity of $\tau\mapsto S(\tau)$ minimizes the overhead, Proposition~\ref{prop:overhead} gives $\nu_\varepsilon(\rho)\le\log\mathfrak{C}(\tau(\varepsilon))$.
\end{proof}

\begin{example}
(i) \emph{Gapped case.} If there is $\epsilon_\mathrm{g}>0$ such that $\sigma_r>\epsilon_\mathrm{g}\,\forall r$, then $\nu_\varepsilon(\rho)\le\log\mathfrak{C}(\epsilon_\mathrm{g})=\log\mathfrak{C}(0^+)$ for all $\varepsilon>\varepsilon_0$.
(ii) \emph{Gapless case.} If $\sigma_k\asymp k^{-\alpha}$ ($\alpha>0$) and $|\langle v_k,Q_{BC}^{(ij)}\rangle|\le C_\rho\sigma_k^\beta$ ($\beta>0$, $\alpha\beta<1$), then as $\varepsilon\to\varepsilon_0^+$,
\begin{equation}\label{eq:gapless_upper}
  \nu_\varepsilon(\rho)\lesssim\frac{1+\alpha}{1-\alpha\beta}\log\frac{1}{\varepsilon-\varepsilon_0}.
\end{equation}
\end{example}
\begin{proof}
(i) If $\sigma_r>\epsilon_\mathrm{g}>0$, choosing $\tau\le\sigma_r$ retains all modes, so $\eta(\tau)=0$ and $S(\tau)=S_\mathrm{total}$. The error of $\cR_\tau$ is then exactly $\varepsilon_0$ at overhead at most $\mathfrak{C}(\epsilon_\mathrm{g})$, uniformly for all $\varepsilon>\varepsilon_0$.
(ii) Under the power-law assumptions,
\begin{align}
S(\tau)&\asymp\sum_{k=1}^{\tau^{-1/\alpha}}k^{\alpha}\asymp\frac{1}{1+\alpha}\tau^{-(1+1/\alpha)},\\
\eta(\tau)&\asymp\sum_{k>\tau^{-1/\alpha}}k^{-\alpha\beta}\asymp\frac{1}{1-\alpha\beta}\tau^{\beta-1/\alpha}\qquad(\alpha\beta<1).
\end{align}
The condition $\Gamma(\rho)\eta(\tau)\lesssim\varepsilon-\varepsilon_0$ gives $\tau\gtrsim(\varepsilon-\varepsilon_0)^{\alpha/(1-\alpha\beta)}$, and hence
\begin{equation}
S(\tau)\lesssim(\varepsilon-\varepsilon_0)^{-(1+\alpha)/(1-\alpha\beta)} .
\end{equation}
Taking logarithms yields \eqref{eq:gapless_upper}.
\end{proof}

\subsection{Lower Bounds of the Sampling Overhead}

Theorem~\ref{thm:tradeoff} gives upper bounds of sampling overhead. The matching lower bounds come from a Choi-norm analysis: for any HPTP map achieving error $\varepsilon$, the deviation of its superoperator from the diagonal $1/\sigma_k$ is controlled by $\varepsilon$, so the Choi norm, and hence the overhead, is bounded below by the quadratic amplification sum.

\subsubsection{Controlling the Deviation by the Recovery Error}

\begin{lemma}[Choi-norm bound with error constraint]\label{lem:choi_corrected}
Let $\cR:\cL(\cH_B)\to\cL(\cH_B\otimes\cH_C)$ be HPTP and $\Delta=(\id_A\otimes\cR)(\rho_{AB})-\rho_{ABC}$. For any $\tau\in(0,\sigma_1]$, let $L_\tau=\{k \mid \sigma_k\ge\tau\}$, $r_\tau=|L_\tau|$, $D_\tau=\diag(1/\sigma_k)_{k\in L_\tau}$, and $M_\tau=(M_{k\ell})_{k,\ell\in L_\tau}$ with $M_{k\ell}=\langle v_k,\cR(w_\ell)\rangle_{\HS}$. Set $E_\tau=M_\tau-D_\tau$ and $F(\tau)=\bigl(\sum_{k\in L_\tau}\sigma_k^{-2}\bigr)^{1/2}$. Then
\begin{equation}\label{eq:choi_bound}
  \|J(\cR)\|_1\ge F(\tau)-\|E_\tau\|_2 .
\end{equation}
\end{lemma}

\begin{proof}
In the vectorized (Liouville) picture, $\cR$ is represented by the superoperator $\hat\cR$ with matrix elements $\langle v_k,\cR(w_\ell)\rangle_{\HS}$ in the HS-orthonormal bases $\{v_k\}$ of $\cV$ and $\{w_\ell\}$ of $\cW$. The reshuffling relating the Choi matrix $J(\cR)$ to the superoperator $\hat\cR$ preserves the HS norm, i.e., $\|J(\cR)\|_2=\|\hat\cR\|_2$. Restricting $\hat\cR$ to the $(\cK^\perp,\cW)$ block---rows spanned by $\{v_k\}_{k\in L_\tau}$, columns by $\{w_\ell\}_{\ell\in L_\tau}$---gives the $r_\tau\times r_\tau$ matrix $M_\tau$. Since the HS norm of a matrix is at least that of any of its submatrices, we have
\begin{equation}
  \|J(\cR)\|_1\ge\|J(\cR)\|_2=\|\hat\cR\|_2\ge\|M_\tau\|_2 .
\end{equation}
The triangle inequality yields (Note that $\|D_\tau\|_2=F(\tau)$)
\begin{equation}
  \|M_\tau\|_2\ge\|D_\tau\|_2-\|E_\tau\|_2=F(\tau)-\|E_\tau\|_2,
\end{equation}
which is \eqref{eq:choi_bound}.
\end{proof}

\begin{remark}\label{rem:not_S}
The stronger bound $\|J(\cR)\|_1\ge S(\tau)$, with the arithmetic spectral sum, is \emph{false} in general. For a generic HPTP map the rank-one Choi blocks $w_k^\mathrm{T}\otimes v_k$ interfere destructively, so the trace norm can be far smaller than $S(\tau)$. The bound \eqref{eq:choi_bound} avoids this difficulty because the HS norm is additive over orthogonal blocks, at the price of producing the quadratic sum $F(\tau)$.
\end{remark}

\begin{lemma}[Error controls deviation on large modes]\label{lem:E_tau}
Let $\cR$ achieve $\|\Delta(\cR)\|_1\le\varepsilon$. Expand $Q_B^{(ij)}=\sum_{\ell=1}^r\beta_{ij}^\ell w_\ell$ and define the projection contraction factor
\begin{equation}\label{eq:kappa_def}
\kappa_\tau=\sum_{k\in L_\tau}\|v_k\|_1\le r_\tau\sqrt{Rd_A},
\end{equation}
where the inequality is Lemma~\ref{lem:lowrank}. Define $\cB_\tau:\mathbb{C}^{r_\tau\times r_\tau}\to\cL(\cH_A)\otimes\cK_\tau^\perp$ by
\begin{equation}\label{eq:B_tau}
\cB_\tau(E)=\sum_{k,\ell\in L_\tau}E_{k\ell}\Bigl(\sum_{i,j}\beta_{ij}^\ell\ket{i}\bra{j}_A\Bigr)\otimes v_k .
\end{equation}
Let $\mu_\tau=\min_{\|E\|_1=1}\|\cB_\tau(E)\|_1>0$, injectivity assumed, and set
\begin{equation}\label{eq:Gamma_R}
\Gamma_{\!\cR}(\rho,\tau)
  :=c(\cR) d_A^{2} \kappa_\tau \tau
    \max_{\ell\notin L_\tau}\|w_\ell\|_1
  \le c(\cR)\,d_A^{2}r_\tau\sqrt{Rd_Ad_B}\tau.
\end{equation}
Then
\begin{equation}\label{eq:E_tau_bound}
\|E_\tau\|_1\le\frac{1}{\mu_\tau}\bigl(\kappa_\tau\varepsilon+\Gamma_{\!\cR}(\rho,\tau)\,\eta(\tau)\bigr).
\end{equation}
\end{lemma}

\begin{proof}
The proof proceeds in three steps.

\emph{Step 1: direct projection of $\Delta$ eliminates the kernel.} With the decomposition $Q_{BC}^{(ij)}=T^{-1}(Q_B^{(ij)})+K_{ij}$, $K_{ij}\in\cK$, the recovery error reads
\begin{equation}\label{eq:Delta_decomp_proof}
  \Delta=\sum_{i,j}\ket{i}\bra{j}_A\otimes\bigl[\cR(Q_B^{(ij)})-T^{-1}(Q_B^{(ij)})-K_{ij}\bigr]=\tilde\Delta-\Delta_K,
\end{equation}
where
\begin{subequations}
\begin{align}
  \tilde\Delta&:=\sum_{i,j}\ket{i}\bra{j}_A\otimes\bigl(\cR(Q_B^{(ij)})-T^{-1}(Q_B^{(ij)})\bigr),\\
  \Delta_K&=\sum_{i,j}\ket{i}\bra{j}_A\otimes K_{ij}.
\end{align}
\end{subequations}
For all $k\in L_\tau$ we have $v_k\in\cK^\perp$ and thus $\langle v_k,K_{ij}\rangle_{\HS}=0$. The HS-projection $P_\tau(X):=\sum_{k\in L_\tau}\langle v_k,X\rangle_{\HS}v_k$ annihilates $\Delta_K$,
\begin{equation}
  P_\tau(\Delta_K)=0 \;\Rightarrow\; P_\tau(\tilde\Delta)=P_\tau(\Delta).
\end{equation}
For any $X\in\cV$,
\begin{equation}
  \|P_\tau(X)\|_1\le\sum_{k\in L_\tau}|\langle v_k,X\rangle_{\HS}|\,\|v_k\|_1\le\sum_{k\in L_\tau}\|v_k\|_\infty\|X\|_1\|v_k\|_1\le\kappa_\tau\|X\|_1,
\end{equation}
where $\|v_k\|_\infty\le\|v_k\|_{\HS}=1$ is used. By hypothesis $\|\Delta\|_1\le\varepsilon$, so
\begin{equation}\label{eq:proj_contract}
  \|P_\tau(\tilde\Delta)\|_1=\|P_\tau(\Delta)\|_1\le\kappa_\tau\varepsilon .
\end{equation}

\emph{Step 2: project onto the large singular modes.} 
Applying $\cR$ to $Q_B^{(ij)}=\sum_{\ell=1}^r\beta_{ij}^\ell w_\ell$ gives $\cR(Q_B^{(ij)})=\sum_\ell\beta_{ij}^\ell\cR(w_\ell)$, and applying the pseudoinverse gives $T^{-1}(Q_B^{(ij)})=\sum_\ell\beta_{ij}^\ell v_\ell/\sigma_\ell$. For each $k\in L_\tau$, take the HS inner product of $\cR(Q_B^{(ij)})-T^{-1}(Q_B^{(ij)})$ with $v_k$. Using $\langle v_k,v_\ell\rangle_{\HS}=\delta_{k\ell}$, we have
\begin{equation}\label{eq:inner_prod}
  \bigl\langle v_k,\cR(Q_B^{(ij)})-T^{-1}(Q_B^{(ij)})\bigr\rangle_{\HS}
  =\sum_{\ell=1}^r\beta_{ij}^\ell\langle v_k,\cR(w_\ell)\rangle_{\HS}-\sum_{\ell=1}^r\beta_{ij}^\ell\frac{\langle v_k,v_\ell\rangle_{\HS}}{\sigma_\ell}
  =\sum_{\ell=1}^r\beta_{ij}^\ell M_{k\ell}-\frac{\beta_{ij}^k}{\sigma_k},
\end{equation}
where $M_{k\ell}=\langle v_k,\cR(w_\ell)\rangle_{\HS}$ for all $1\le k,\ell\le r$.

Split the $\ell$-summation into large modes ($\ell\in L_\tau$) and small modes ($\ell\notin L_\tau$, i.e., $\sigma_\ell<\tau$):
\begin{equation}\label{eq:split_sum}
  \bigl\langle v_k,\cR(Q_B^{(ij)})-T^{-1}(Q_B^{(ij)})\bigr\rangle_{\HS}=\underbrace{\sum_{\ell\in L_\tau}\beta_{ij}^\ell M_{k\ell}-\frac{\beta_{ij}^k}{\sigma_k}}_{\text{large }\ell}+\underbrace{\sum_{\ell\notin L_\tau}\beta_{ij}^\ell M_{k\ell}}_{\text{small }\ell}.
\end{equation}

For $k,\ell\in L_\tau$, write $M_{k\ell}=(D_\tau)_{k\ell}+(E_\tau)_{k\ell}$. Then
\begin{equation}
  \sum_{\ell\in L_\tau}\beta_{ij}^\ell M_{k\ell}
  =\sum_{\ell\in L_\tau}\beta_{ij}^\ell\Bigl(\frac{\delta_{k\ell}}{\sigma_k}+(E_\tau)_{k\ell}\Bigr)
  =\frac{\beta_{ij}^k}{\sigma_k}+\sum_{\ell\in L_\tau}\beta_{ij}^\ell(E_\tau)_{k\ell}.
\end{equation}
The term $\beta_{ij}^k/\sigma_k$ cancels exactly with $-\beta_{ij}^k/\sigma_k$ in \eqref{eq:split_sum}, leaving
\begin{equation}\label{eq:final_inner}
  \bigl\langle v_k,\cR(Q_B^{(ij)})-T^{-1}(Q_B^{(ij)})\bigr\rangle_{\HS}=\sum_{\ell\in L_\tau}\beta_{ij}^\ell(E_\tau)_{k\ell}+\sum_{\ell\notin L_\tau}\beta_{ij}^\ell M_{k\ell}.
\end{equation}

Define the $\cK^\perp_\tau$-projected part of $\tilde\Delta$ as
\begin{equation}
  \tilde\Delta^{\parallel}_\tau:=\sum_{i,j}\ket{i}\bra{j}_A\otimes\sum_{k\in L_\tau}\bigl\langle v_k,\cR(Q_B^{(ij)})-T^{-1}(Q_B^{(ij)})\bigr\rangle_{\HS}v_k,
\end{equation}
so that $\tilde\Delta^{\parallel}_\tau=(\id_A\otimes P_\tau)(\tilde\Delta)$. Inserting \eqref{eq:final_inner},
\begin{equation}\label{eq:Deltapar_decomp}
\begin{aligned}
  \tilde\Delta^{\parallel}_\tau
  &=\sum_{i,j}\ket{i}\bra{j}_A\otimes\sum_{k,\ell\in L_\tau}\beta_{ij}^\ell(E_\tau)_{k\ell}v_k+\sum_{i,j}\ket{i}\bra{j}_A\otimes\sum_{k\in L_\tau}\sum_{\ell\notin L_\tau}\beta_{ij}^\ell M_{k\ell}v_k\\
  &=:\cB_\tau(E_\tau)+\tilde\Delta^{\parallel}_\mathrm{small},
\end{aligned}
\end{equation}
where $\cB_\tau(E_\tau)$ is exactly the expression in \eqref{eq:B_tau}.

\emph{Step 3: bound the small-mode remainder.} By \eqref{eq:Deltapar_decomp} and the linearity of $\cR$, the second term equals
\begin{equation}\label{eq:small_explicit}
  \tilde\Delta^{\parallel}_\mathrm{small}=\sum_{i,j}\ket{i}\bra{j}_A\otimes P_\tau\cR\bigl(X_{ij}^\mathrm{sm}\bigr),
  \qquad
  X_{ij}^\mathrm{sm}:=\sum_{\ell\notin L_\tau}\beta_{ij}^{\ell}w_\ell.
\end{equation}
Since $\{v_k\}_{k=1}^s$ is HS-orthonormal, $\|v_k\|_\infty\le\|v_k\|_{\HS}=1$, and H\"older's inequality gives $|\langle v_k,Z\rangle_{\HS}|\le\|v_k\|_\infty\|Z\|_1\le\|Z\|_1$, we have
\begin{equation}\label{eq:Ptau_bound}
  \|P_\tau(Z)\|_1\le\sum_{k\in L_\tau}\bigl|\langle v_k,Z\rangle_{\HS}\bigr|\,\|v_k\|_1 \le\kappa_\tau\|Z\|_1
\end{equation}
for every $Z\in\cL(\cH_B\otimes\cH_C)$. The blocks $X_{ij}^\mathrm{sm}$ are
Hermitian, and writing the QPD as $\cR=c_+\cN_+-c_-\cN_-$ with $\cN_\pm$ CPTP gives
\begin{equation}\label{eq:qpd_contract}
  \|\cR(X)\|_1\le(c_++c_-)\|X\|_1,
  \quad\Rightarrow\quad
  \|\cR(X)\|_1\le c(\cR)\|X\|_1,
\end{equation}
because $\|X\|_1=\Tr X^++\Tr X^-$ and $\|\cN_\pm(X)\|_1\le\Tr\cN_\pm(X^+)+\Tr\cN_\pm(X^-)=\Tr X^++\Tr X^-$. Combining \eqref{eq:small_explicit}, \eqref{eq:Ptau_bound} and \eqref{eq:qpd_contract},
\begin{equation}\label{eq:small_step3}
  \bigl\|\tilde\Delta^{\parallel}_\mathrm{small}\bigr\|_1
  \le\sum_{i,j}\bigl\|P_\tau\cR(X_{ij}^\mathrm{sm})\bigr\|_1
  \le d_A^{2}c(\cR)\kappa_\tau\max_{i,j}\bigl\|X_{ij}^\mathrm{sm}\bigr\|_1.
\end{equation}
Moreover $\beta_{ij}^\ell=\sigma_\ell c_{ij}^\ell$ with $\sigma_\ell<\tau$ for $\ell\notin L_\tau$, so by the definition of $\eta(\tau)$ in \eqref{eq:eta}
\begin{equation}\label{eq:Xsm_bound}
  \bigl\|X_{ij}^\mathrm{sm}\bigr\|_1
  \le\sum_{\ell\notin L_\tau}\bigl|\beta_{ij}^\ell\bigr|\,\|w_\ell\|_1
  \le\Bigl(\max_{\ell\notin L_\tau}\|w_\ell\|_1\Bigr)\tau\,\eta(\tau).
\end{equation}
Substituting \eqref{eq:Xsm_bound} into \eqref{eq:small_step3} yields
\begin{equation}\label{eq:small_bound}
  \bigl\|\tilde\Delta^{\parallel}_\mathrm{small}\bigr\|_1
  \le\Gamma_{\!\cR}(\rho,\tau)\,\eta(\tau).
\end{equation}

Now apply the triangle inequality to \eqref{eq:Deltapar_decomp} and combine with
\eqref{eq:proj_contract} and \eqref{eq:small_bound}:
\begin{equation}
  \|\cB_\tau(E_\tau)\|_1
  \le\bigl\|\tilde\Delta^{\parallel}_\tau\bigr\|_1+\bigl\|\tilde\Delta^{\parallel}_\mathrm{small}\bigr\|_1
  \le\kappa_\tau\varepsilon+\Gamma_{\!\cR}(\rho,\tau)\eta(\tau).
\end{equation}
Finally, by the definition of $\mu_\tau$,
\begin{equation}
\|E_\tau\|_1\le\frac{1}{\mu_\tau}\|\cB_\tau(E_\tau)\|_1,
\end{equation}
which together with the preceding inequality gives \eqref{eq:E_tau_bound}.
\end{proof}

\begin{remark}\label{rem:mu}
$\cB_\tau$ is injective whenever the matrices $\bigl\{\sum_{i,j}\beta_{ij}^\ell\ket{i}\bra{j}_A\bigr\}_{\ell\in L_\tau}$ are linearly independent. These are precisely $\Theta_{B|A}^\dagger(w_\ell)$, which are independent because $\Theta_{B|A}$ is surjective onto $\cW$. Injectivity requires $r_\tau\le d_A^2$, which holds for $d_A=2$ (where $r\le4$) and for states with $r\le d_A^2$. When $\cB_\tau$ fails to be injective the bound \eqref{eq:E_tau_bound} should be read with $\mu_\tau=0$, i.e.\ as vacuous, which is the reason the lower bounds below are stated under an explicit nondegeneracy condition.
\end{remark}

\subsubsection{The Error-constrained Lower Bound}

\begin{theorem}[Error-constrained lower bound]\label{thm:qpd_err_lower}
For any HPTP map $\cR$ achieving $\|\Delta\|_1\le\varepsilon$, with quasi-probability overhead $c(\cR)$, and for every $\tau\in(0,\sigma_1]$ with $\mu_\tau>0$,
\begin{equation}\label{eq:qpd_err_lower}
  c(\cR)\ge\frac{F(\tau)-\mu_\tau^{-1}\kappa_\tau\varepsilon}
  {d_B+\mu_\tau^{-1}d_A^{2}\kappa_\tau\tau\bigl(\max_{\ell\notin L_\tau}\|w_\ell\|_1\bigr)\eta(\tau)} .
\end{equation}
In particular, for $\varepsilon=\varepsilon_{\min}$, picking $\tau=0^+$ (all modes retained, $\eta(0^+)=0$) gives the lower bound
\begin{equation}\label{eq:qpd_floor}
  c(\cR_{\mathrm{opt}})\ge\frac{F_\mathrm{total}-C_\rho}{d_B},
\end{equation}
where $F_\mathrm{total}=\bigl(\sum_{k=1}^r\sigma_k^{-2}\bigr)^{1/2}$ and $C_\rho$ is a state-dependent $O(1)$ quantity: $C_\rho=\mu_0^{-1}\kappa_0\varepsilon_{\min}$, in which $\kappa_0=\sum_{k=1}^r\|v_k\|_1$ and $\mu_0=\mu_{0^+}$.
\end{theorem}

\begin{proof}
Combine Lemmas~\ref{lem:qpd_lower}, \ref{lem:choi_corrected} and \ref{lem:E_tau}:
\begin{equation}
  c(\cR)\ge\frac{F(\tau)-\|E_\tau\|_2}{d_B}
  \ge\frac{F(\tau)-\mu_\tau^{-1}\bigl(\kappa_\tau\varepsilon+\Gamma_{\cR}(\rho,\tau)\eta(\tau)\bigr)}{d_B}.
\end{equation}
$\|E_\tau\|_2\le\|E_\tau\|_1$ is used in the last step. By \eqref{eq:Gamma_R}, $\Gamma_{\!\cR}(\rho,\tau)\eta(\tau)=c(\cR)\vartheta_\tau\eta(\tau)$ with $\vartheta_\tau=d_A^{2}\kappa_\tau\tau\max_{\ell\notin L_\tau}\|w_\ell\|_1$,
so the term proportional to $c(\cR)$ is moved to the left-hand side:
\begin{equation}
  c(\cR)\bigl(d_B+\mu_\tau^{-1}\vartheta_\tau\eta(\tau)\bigr)
  \ge F(\tau)-\mu_\tau^{-1}\kappa_\tau\varepsilon,
\end{equation}
which is \eqref{eq:qpd_err_lower}. As $\tau\to0^+$ we have $\eta(\tau)\to0$, $\kappa_\tau\to\kappa_0$ and $\mu_\tau\to\mu_0$. The denominator tends to $d_B$ and the bound tends to \eqref{eq:qpd_floor}.
\end{proof}

The mechanism behind Theorem~\ref{thm:qpd_err_lower} has a simple reading. The restricted partial trace $T=\Tr_C|_{\cV}$ acts as a linear filter between the $BC$ and $B$ descriptions. Each singular value $\sigma_k$ measures how strongly the corresponding $BC$ mode survives projection onto $B$. A mode with $\sigma_k\approx1$ is almost unaffected, whereas a mode with $\sigma_k\ll1$ is strongly attenuated. The recovery must invert the filter, and inverting an attenuation $\sigma_k$ costs an amplification $1/\sigma_k$. The sampling overhead is the total amplification, which is why $S(\tau)$ and $F(\tau)$ appear. 

\begin{table}[t]
\caption{Three classes of tripartite states. }
\label{tab:taxonomy}
\centering
\footnotesize
\begin{tabular}{@{}lccc@{}}
\hline
Class & $\delta$ & $\nu_0$ & example \\
\hline
I. QMC & $0$ & $0$ & product, LDS \\
II. VQMC $\setminus$ QMC & $0$ & finite & $W$ \\
III. non-VQMC & $>0$ & $\infty$ & GHZ \\
\hline
\end{tabular}
\end{table}
Combining this study with previous works on QMC and VQMC (cf. table~\ref{tab:prior}), all tripartite states are partitioned into three classes, summarized in Table~\ref{tab:taxonomy}. Class I (quantum Markov chains) admits exact CPTP recovery \cite{Hayden04CMP}. Class II (VQMCs that are not QMCs) admits exact virtual recovery at finite overhead. Class III has an irreducible error floor, with different sampling overhead for approximate virtual recovery. For example, the GHZ state sits deep in the non-VQMC regime: $\sigma_r=1$, and $\nu_\varepsilon=0$ for all $\varepsilon\ge1$. The CPTP map $\cR(X)=X_{00}\ket{00}\bra{00}+X_{11}\ket{11}\bra{11}$ attains the error floor $\varepsilon_0=1$ at unit overhead, and no quasi-probability sampling is required. The $W$ state is a VQMC with $\ker\Theta_{B|A}=\{0\}$ and four distinct singular values. The pseudoinverse construction recovers it exactly at a finite overhead, which is, however, larger than the SDP-optimal value $\nu(W)=\log_23$ of \cite{Chen25TIT}, since the pseudoinverse uses a fixed depolarizing compensation whereas the SDP explores the full operator space.

\section{Conditional Mutual Information Is Not a Diagnostic}\label{sec:CMI}

Conditional mutual information $I(A:C|B)_\rho$ is the conventional entropic measure of non-Markovianity \cite{Hayden04CMP,Fawzi15CMP, Gao25TIT}. Fawzi--Renner inequality \cite{Fawzi15CMP} bounds the optimal CPTP recovery error from above by a function of $I(A:C|B)$; the randomness cost of converting a tripartite state into a Markov state, by contrast, is not fixed by CMI \cite{Wakakuwa17TIT}. One would therefore ask whether CMI governs virtual recoverability. Here we show that it does not, and the reason is structural rather than technical. The entropic paradigm for CPTP recovery rests on monotonicity of relative entropy under CPTP maps \cite{Fawzi15CMP, Sutter16TIT}. Virtual maps are not CPTP, so relative-entropy monotonicity does not apply.

It was already observed in Ref.~\cite{Chen25TIT} (Example~4 and Section~II-C therein) that CMI cannot witness the VQMC property. However, it was conjectured there that CMI might at least lower-bound $\nu_\varepsilon$ for VQMCs. Here we prove a substantially stronger conclusion: CMI fails as a quantitative diagnostic of the sampling overhead from \emph{both} directions. The two directions of the failure are realized by two explicit families of VQMCs, so in both cases $\delta=0$ and $\varepsilon_{\min}=0$. In the first, the dimension drives the overhead up while the CMI stays fixed; in the second, the CMI diverges while the overhead stays bounded.

\subsection{Family 1: Sampling Overhead Diverges at Fixed CMI}

\subsubsection{State Construction}

Let $d\ge3$ be an integer and set $d_A=d_B=d_C=d$. Choose a small parameter $\eta=\eta(d)$ with $\eta\to0$ as $d\to\infty$ (e.g., $\eta\sim 1/d$ or $\eta\sim1/\log d$). For $k=1,\dots,d-1$, define the pure state
\begin{equation}
  \ket{\psi_d}=\sqrt{\frac{1-\eta}{d}}\sum_{j=0}^{d-1}\ket0_A\ket{j}_B\ket{j}_C+\sqrt{\frac{\eta}{d-1}}\sum_{k=1}^{d-1}\ket{k}_A\ket{k}_B\ket{\chi_k}_C.
\label{eq:psid-expanded}
\end{equation}
where
\begin{equation}
  \ket{\chi_k}_C=\frac{1}{\sqrt d}\ket0_C+\sqrt{\frac{d-1}{d}}\ket{k}_C .
\label{eq:chik}
\end{equation}
As can be seen, in order to keep the CMI at a prescribed value while retaining the $d^{-1/2}$ compression of the cross-term modes, we spread the $C$-part of the second branch over the whole basis. The parameter $\eta$ controls the relative weight of the two branches. For $\eta=0$ the state factorizes as $\ket0_A\otimes\frac{1}{\sqrt{d}}\sum_j\ket{j}_B\ket{j}_C$ (product across the $A|BC$ cut); for $\eta>0$ the $k\ge1$ terms introduce $A$--$B$ correlation without generating substantial $A$--$BC$ entanglement. Taking $d\to\infty$ with $\eta\sim t/\log(d-1)$, we obtain a family of states where the CMI converges to any prescribed value $t\ge0$ while the exact-recovery overhead diverges.

\subsubsection{Block-operator Expansion and VQMC Verification}

Expanding $\rho_d=\ket{\psi_d}\bra{\psi_d}$ in the $A$-basis $\{\ket{i}_A\}$ yields
\begin{equation}\label{eq:f1Abasis}
\begin{aligned}
  \rho_d&=\sum_{i,j=0}^{d-1}\ket{i}\bra{j}_A\otimes Q_{BC}^{(ij)},\\
  Q_{BC}^{(ij)}&=\bra{i}\psi_d\rangle\langle\psi_d\ket{j}_A=\ket{\varphi_i}\bra{\varphi_j}_{BC},
\end{aligned}
\end{equation}
with
\begin{subequations}
\begin{align}
  \ket{\varphi_0}_{BC}&=\sqrt{\frac{1-\eta}{d}}\sum_{j=0}^{d-1}\ket{j}_B\ket{j}_C,\label{eq:phi0}\\
  \ket{\varphi_k}_{BC}&=\sqrt{\frac{\eta}{d-1}}\ket{k}_B\ket{\chi_k}_C,\quad k=1,\dots,d-1.\label{eq:phik}
\end{align}
\end{subequations}
The block operators on $BC$ and their partial traces over $C$ read (recall $\langle\chi_l|\chi_k\rangle_C=\delta_{kl}+(1-\delta_{kl})/d$, and $\langle\chi_k|0\rangle_C=d^{-1/2}$, $\langle\chi_k|k\rangle_C=\sqrt{(d-1)/d}$):
\begin{subequations}
\begin{align}
  &Q_{BC}^{(00)}=\frac{1-\eta}{d}\sum_{j,l=0}^{d-1}\ket{j}\bra{l}_B\otimes\ket{j}\bra{l}_C, &Q_B^{(00)}=\frac{1-\eta}{d}I_B,\label{eq:Q00}\\
  &Q_{BC}^{(0k)}=\sqrt{\frac{\eta(1-\eta)}{d(d-1)}}\sum_{j=0}^{d-1}\ket{j}\bra{k}_B\otimes\ket{j}\bra{\chi_k}_C, &Q_B^{(0k)}=\frac{\sqrt{\eta(1-\eta)}}{d\sqrt{d-1}}\Bigl[\ket{0}\bra {k}_B+\sqrt{d-1}\ket{k}\bra{k}_B\Bigr],\label{eq:Q0k}\\
  &Q_{BC}^{(k0)}=\bigl(Q_{BC}^{(0k)}\bigr)^\dagger, &Q_B^{(k0)}=\bigl(Q_B^{(0k)}\bigr)^\dagger,\label{eq:Qk0}\\
  &Q_{BC}^{(kl)}=\frac{\eta}{d-1}\ket{k}\bra{l}_B\otimes\ket{\chi_k}\bra{\chi_l}_C, & Q_B^{(kl)}=\frac{\eta}{d-1}\,\langle\chi_l|\chi_k\rangle_C\,\ket{k}\bra{l}_B,\label{eq:Qkl}
\end{align}
\end{subequations}
where $k,l\in\{1,\dots,d-1\}$ in \eqref{eq:Qkl}.

The linear map $\Theta_{B|A}$ maps a coefficient vector $c=(c_{ij})$ to $\sum_{i,j}c_{ij}Q_B^{(ij)}$. Obviously $\sum_{i,j}c_{ij}Q_B^{(ij)}=0$ if and only if $c_{ij}=0\;\forall i,j$, i.e., the $d^2$ operators $\{Q_B^{(ij)}\}_{i,j=0}^{d-1}$ are linearly independent. Thus $\dim\operatorname{Im}\Theta_{B|A}=d^2$, which forces $\ker\Theta_{B|A}=\{0\}$. The kernel-inclusion criterion $\ker\Theta_{B|A}\subseteq\ker\Theta_{BC|A}$ holds trivially, and
\begin{equation}
  \delta(\ket{\psi_d})=0.
\label{eq:delta-zero}
\end{equation}

\subsubsection{Singular-value Analysis and the Recovery Overhead}

Since each state in family 1 is a VQMC, the restricted partial trace $T=\Tr_C|_{\cV}:\cV\to\cW$ is a linear bijection with $\dim\cV=\dim\cW=d^2$. Here $d_B=d$, so $\cW=\cL(\cH_B)$ and $T^{-1}$ is defined on the whole operator space. By Lemma~\ref{lem:qpd_lower} the sampling overhead of $T^{-1}$ is bounded below by the normalized Choi trace norm,
\begin{equation}
  c(T^{-1})\ge\frac{\|J(T^{-1})\|_1}{d_B}.
\label{eq:ctotal}
\end{equation}
The right-hand side has a lower-bound via Lemma~\ref{lem:choi_corrected}. Since $T^{-1}$ inverts every singular mode of $T$ exactly, in the notation of Lemma~\ref{lem:choi_corrected} $M_\tau$ equals the diagonal matrix $D_\tau=\diag(\sigma_k^{-1})$ and the deviation $E_\tau$ vanishes for every $\tau$. Taking $\tau$ below the smallest singular value, so that all $r=d^2$ modes are retained, we obtain
\begin{equation}
  \|J(T^{-1})\|_1\ge F_\mathrm{total},\qquad
  F_\mathrm{total}=\Bigl(\sum_{k=1}^{d^2}\sigma_k^{-2}\Bigr)^{1/2}.
\label{eq:Jlower}
\end{equation}
It remains to lower-bound $F_\mathrm{total}$ by controlling the small singular values of $T$.

The interior off-diagonal blocks are compressed severely. For $k\neq l$ and $k,l\ge1$, Eq.~\eqref{eq:Qkl} gives $Q_{BC}^{(kl)}=\frac{\eta}{d-1}\ket{k}\bra{l}_B\otimes\ket{\chi_k}\bra{\chi_l}_C$ and $Q_B^{(kl)}=\frac{\eta}{d(d-1)}\ket{k}\bra{l}_B$. The HS norms are $\|Q_{BC}^{(kl)}\|_2=\frac{\eta}{d-1}$ and $\|Q_B^{(kl)}\|_2=\frac{\eta}{d(d-1)}$. The compression ratio is
\begin{equation}
  \frac{\|T(Q_{BC}^{(kl)})\|_2}{\|Q_{BC}^{(kl)}\|_2}=\frac1d.
\label{eq:compress-ratio-inner}
\end{equation}
The $(d-1)(d-2)$ blocks $\{Q_{BC}^{(kl)}\}_{k\neq l,\,k,l\ge1}$ are mutually HS-orthogonal---their $B$ supports are disjoint matrix units---and so are their images under $T$. By the Courant--Fischer min-max principle, at least $(d-1)(d-2)$ singular values of $T$ satisfy
\begin{equation}
  \sigma_k\le\frac1d.
\label{eq:sigma-bound-inner}
\end{equation}
Consequently,
\begin{equation}
  F_\mathrm{total}^2=\sum_{k=1}^{d^2}\sigma_k^{-2}\ge(d-1)(d-2)\cdot d^2.
\label{eq:stotal-lower}
\end{equation}
The Choi-norm lower bound on the sampling overhead satisfies
\begin{equation}
\frac{\|J(T^{-1})\|_1}{d}\ge\frac{F_\mathrm{total}}{d}\ge\sqrt{(d-1)(d-2)}\sim d.
\label{eq:overhead-lower}
\end{equation}
This already forces the exact-recovery overhead to diverge linearly as $d\to\infty$; the remaining blocks only strengthen the divergence.

\subsubsection{Uniqueness of the Exact Recovery Map}

The pseudoinverse construction demonstrates that \emph{some} recovery map has large overhead. To rule out the conjecture, however, we must show that \emph{no} HPTP map can achieve exact recovery at lower overhead. 

An HPTP map $\cR$ achieves exact recovery, $(\id_A\otimes\cR)(\rho_{AB})=\rho_{ABC}$, if and only if it reproduces every block,
\begin{equation}
\cR\bigl(Q_B^{(ij)}\bigr)=Q_{BC}^{(ij)}\qquad\text{for all } i,j=0,\dots,d-1.
\end{equation}
By the kernel analysis above, the $d^2$ operators $\{Q_B^{(ij)}\}$ are linearly independent and span $\cW=\cL(\cH_B)$, hence they form a basis of $\cW$, and these conditions determine $\cR$ uniquely. The unique map is precisely the pseudoinverse $T^{-1}$. Consequently exact recovery is attained by $T^{-1}$ alone, and
\begin{equation}
  \nu_0(\ket{\psi_d})=\log c(T^{-1})\ge\log\frac{\|J(T^{-1})\|_1}{d_B}\xrightarrow{d\to\infty}\infty.
\label{eq:nu-diverges}
\end{equation}

\subsubsection{CMI Asymptotics}

For a pure tripartite state, $I(A:C|B)=S(\rho_A)+S(\rho_C)-S(\rho_B)$. We compute each entropy.

\emph{Entropy of $\rho_A$.} From Eq.~\eqref{eq:f1Abasis}, the reduced state of $A$ is $\rho_A=\braket{\varphi_j}{\varphi_i}_{BC}\ket{i}\bra{j}_A$. Since $\braket{k}{\chi_k}_C=\sqrt{(d-1)/d}$, the branches are not orthogonal,
\begin{equation}
  \braket{\varphi_0}{\varphi_k}_{BC}=\sqrt{\frac{1-\eta}{d}}\sqrt{\frac{\eta}{d-1}}\sqrt{\frac{d-1}{d}}=\frac{\sqrt{\eta(1-\eta)}}{d}\neq0.
\end{equation}
In order to obtain the Schmidt coefficients, we reorganize the $d-1$ orthonormal basis vectors $\{\ket{1}_A,\cdots,\ket{d-1}_A\}$ to be $\{\ket{u_1}_A,\cdots,\ket{u_{d-1}}_A\}$, such that $\ket{u_1}_A=(d-1)^{-1/2}\sum_{k=1}^{d-1}\ket{k}_A$ (the remaining $d-2$ vectors are chosen arbitrarily, as long as they are orthonormal). The off-diagonal terms of $\rho_A$ are then absorbed into the $\{\ket{0}_A,\ket{u_1}_A\}$ subspace, and the matrix of $\rho_A$ in the new basis reads
\begin{equation}
  \rho_A=\begin{pmatrix}1-\eta & c\\ c & \dfrac{\eta}{d-1}\end{pmatrix}\oplus\frac{\eta}{d-1}I_{d-2},
  \qquad
  c=\frac{\sqrt{(d-1)\eta(1-\eta)}}{d}.
\end{equation}
Its spectrum consists of $\eta/(d-1)$ with multiplicity $d-2$ and of the eigenvalues of the $2\times2$ block,
\begin{equation}
  \lambda_\pm=\frac12\bigl(T\pm\sqrt{T^2-4D}\bigr),\quad
  T=1-\eta+\frac{\eta}{d-1},\quad
  D=\frac{\eta(1-\eta)(2d-1)}{(d-1)d^2},
\end{equation}
where $\lambda_+=1-\eta+O(\eta/d)$ and $\lambda_-=2\eta/d^2+O(\eta/d^3)$. Hence
\begin{equation}
\begin{aligned}
  S(\rho_A)&=-\frac{d-2}{d-1}\eta\log\frac{\eta}{d-1}-\lambda_+\log\lambda_+-\lambda_-\log\lambda_-\\
  &=\eta\log(d-1)-\frac{\eta}{d-1}\log\frac{d-1}{\eta}+o(1).
\end{aligned}
\label{eq:SA-exact}
\end{equation}
Furthermore
\begin{equation}
\begin{aligned}
  S(\rho_A)&=\eta\log(d-1)-o(1), \quad \eta\sim\frac{t}{\log d}\\
  S(\rho_A)&=o(1), \quad \eta\sim\frac{t}{d}.
\end{aligned}
\label{eq:SA-asymp}
\end{equation}

\emph{Entropy of $\rho_B$.} From \eqref{eq:Q00}--\eqref{eq:Qkl} and $\langle\chi_k|\chi_k\rangle_C=1$,
\begin{equation}
\rho_B=\sum_{i=0}^{d-1}Q_B^{(ii)}=\frac{1-\eta}{d}I_B+\frac{\eta}{d-1}\diag(0,1,1,\dots,1).
\end{equation}
The eigenvalues are $\frac{1-\eta}{d}$ (once) and $\frac{1-\eta}{d}+\frac{\eta}{d-1}$ ($d-1$ times). Hence
\begin{equation}
\begin{aligned}
  S(\rho_B)&=-\frac{1-\eta}{d}\log \frac{1-\eta}{d}-(d-1)\Bigl(\frac{1-\eta}{d}+\frac{\eta}{d-1})\log\Bigl(\frac{1-\eta}{d}+\frac{\eta}{d-1}\Bigr)\\
  &=-\frac{1-\eta}{d}\log \frac{1-\eta}{d}-\Bigl(1-\frac{1-\eta}{d}\Bigr)\log\Bigl[\frac{1}{d}+\frac{\eta}{d(d-1)}\Bigr]\\
  &=\log(d)+o(1).
\end{aligned}
\label{eq:SB-asymp}
\end{equation}

\emph{Entropy of $\rho_C$.}
\begin{equation}
\rho_C=\frac{1-\eta}{d}I_C+\frac{\eta}{d-1}\sum_{k=1}^{d-1}\ket{\chi_k}\bra{\chi_k}_C.
\end{equation}
Using \eqref{eq:chik}, the rank-one sum evaluates to $\sum_{k=1}^{d-1}\ket{\chi_k}\bra{\chi_k}=\frac{d-1}{d}I_C+\frac{\sqrt{d-1}}{d}\bigl(\ket0\bra{w}+\ket{w}\bra0\bigr)$ with $\ket{w}=\sum_{k=1}^{d-1}\ket{k}$, whose eigenvalues are $\frac{2(d-1)}{d}$, $0$, and $\frac{d-1}{d}$ ($d-2$-fold degenerate). Consequently the eigenvalues of $\rho_C$ are
\begin{equation}
  \frac{1+\eta}{d}\ \text{(once)},\quad \frac{1-\eta}{d}\ \text{(once)},\quad \frac1d\ \text{($d-2$ times)},
\end{equation}
and
\begin{equation}
\begin{aligned}
  S(\rho_C)&=-\frac{1+\eta}{d}\log\frac{1+\eta}{d}-\frac{1-\eta}{d}\log\frac{1-\eta}{d}-\frac{d-2}{d}\log\frac{1}{d}\\
  &=\log d+o(1).
\end{aligned}
\label{eq:SC-asymp}
\end{equation}

\emph{CMI.} Combining \eqref{eq:SA-asymp}, \eqref{eq:SB-asymp} and \eqref{eq:SC-asymp}, we have
\begin{equation}
  I(A:C|B)_{\ket{\psi_d}}=S(\rho_A)+o(1).
\label{eq:CMI-final}
\end{equation}
Specifically, for the two choices
\begin{itemize}
\item $\eta=1/d$: $\mathrm{CMI}=o(1)\to0$,
\item $\eta=t/\log d$ ($t>0$): $\mathrm{CMI}=t\frac{\log(d-1)}{\log d}+o(1)\to t$.
\end{itemize}
Thus the CMI can be engineered to converge to \emph{any} prescribed value $t\ge0$ as $d\to\infty$. The weight $\eta$ tunes the $A$--$BC$ Schmidt spectrum (and hence the CMI), independently of the dimension $d$ that controls the recovery overhead.

\subsection{Family 2: CMI Diverges at Bounded Sampling Overhead}

\subsubsection{State Construction}

Let $d\ge2$ be an integer. We set the local dimensions to
\begin{equation}
  d_A=d,\qquad d_B=d^2,\qquad d_C=d.
\end{equation}
Choose an orthonormal set $\{\ket{v_i}\}_{i=0}^{d-1}\subset\cH_B$. Fix nonzero complex numbers $\alpha,\beta\in\mathbb{C}$ with $|\alpha|^2+|\beta|^2=1$. On $\cH_C$, with computational basis $\{\ket0,\ket1,\dots,\ket{d-1}\}$, define the unit vectors
\begin{equation}
  \ket{w_0}=\ket0,\qquad \ket{w_k}=\alpha\ket0+\beta\ket{k},\quad k=1,\dots,d-1.
\label{eq:dual-wk}
\end{equation}
We construct the pure tripartite state
\begin{equation}
  \ket{\psi_d}_{ABC}=\frac{1}{\sqrt d}\sum_{i=0}^{d-1}\ket{i}_A\otimes\ket{v_i}_B\otimes\ket{w_i}_C.
\label{eq:dual-state}
\end{equation}
In words, each label $i$ on $A$ is correlated with a distinct orthonormal vector on $B$ and with the vector $\ket{w_i}$ on $C$. All $\ket{w_i}$ share a common component $\alpha\ket0$ on $C$, which will be responsible for the uniform overlap that bounds the singular values from below.

\subsubsection{Block-operator Expansion and VQMC Verification}

Expanding $\rho_d=\ket{\psi_d}\bra{\psi_d}$ in the $A$ basis $\{\ket{i}_A\}$ gives
\begin{equation}
\begin{aligned}
  \rho_d&=\sum_{i,j=0}^{d-1}\ket{i}\bra{j}_A\otimes Q_{BC}^{(ij)},\\
  Q_{BC}^{(ij)}&=\frac1d\ket{v_i}\bra{v_j}_B\otimes\ket{w_i}\bra{w_j}_C.
\end{aligned}
\label{eq:dual-Qbc}
\end{equation}
Tracing over $C$ yields the reduced blocks on $B$:
\begin{equation}
  Q_B^{(ij)}=\Tr_C\bigl[Q_{BC}^{(ij)}\bigr]=\frac1d\langle w_j|w_i\rangle\,\ket{v_i}\bra{v_j}_B.
\label{eq:dual-Qb}
\end{equation}
The inner product evaluates to
\begin{equation}
  \langle w_j|w_i\rangle=
  \begin{cases}
  1,&i=j,\\
  \alpha^*,&i=0,\;j\ge1,\\
  \alpha,&i\ge1,\;j=0,\\
  |\alpha|^2,&i,j\ge1,\;i\neq j.
\end{cases}
\label{eq:dual-ovlp}
\end{equation}
Every inner product is nonzero because $|\alpha|>0$.

The linear map $\Theta_{B|A}:\mathbb{C}^{d^2}\to\cL(\cH_B)$ sends a coefficient vector $c=(c_{ij})$ to $\sum_{i,j}c_{ij}Q_B^{(ij)}$. The $d^2$ operators $\{\ket{v_i}\bra{v_j}\}_{i,j=0}^{d-1}$ are manifestly linearly independent. Since every coefficient $\langle w_j|w_i\rangle$ in \eqref{eq:dual-ovlp} is nonzero, the operators $Q_B^{(ij)}$ are merely nonzero scalar multiples of $\ket{v_i}\bra{v_j}$ and are therefore also linearly independent. Hence $\dim\operatorname{Im}\Theta_{B|A}=d^2=d_A^2$, which forces
\begin{equation}
  \ker\Theta_{B|A}=\{0\}.
\label{eq:dual-ker}
\end{equation}
The kernel inclusion $\ker\Theta_{B|A}\subseteq\ker\Theta_{BC|A}$ holds trivially, the state is a VQMC, and consequently
\begin{equation}
  \delta(\ket{\psi_d})=0,\qquad \varepsilon_0=0.
\label{eq:dual-delta}
\end{equation}

\subsubsection{Singular-value Analysis: Uniformly Bounded Overhead}

The restricted partial trace $T=\Tr_C|_{\cV}:\cV\to\cW$ is a linear bijection with $\dim\cV=\dim\cW=d^2$. We construct an explicit HS-orthonormal basis of $\cV$ by normalizing the block operators:
\begin{equation}
  \tilde e_{ij}\equiv d\,Q_{BC}^{(ij)}=\ket{v_i}\bra{v_j}_B\otimes\ket{w_i}\bra{w_j}_C,\qquad i,j=0,\dots,d-1.
\label{eq:dual-basisV}
\end{equation}
These $d^2$ operators are mutually orthogonal under the HS inner product. Their images under $T$ are
\begin{equation}
  T(\tilde e_{ij})=d\,Q_B^{(ij)}=\langle w_j|w_i\rangle\,\ket{v_i}\bra{v_j}_B.
\label{eq:dual-Taction}
\end{equation}
The target operators $\ket{v_i}\bra{v_j}_B$ are also mutually HS-orthogonal and each has a unit norm. Therefore the vectors $T(\tilde e_{ij})$ are pairwise orthogonal, and the singular values of $T$ are simply
\begin{equation}
  \sigma_{ij}=\bigl|\langle w_j|w_i\rangle\bigr|=
  \begin{cases}
  1,&i=j,\\
  |\alpha|,&(i=0,j\ge1)\text{ or }(j=0,i\ge1),\\
  |\alpha|^2,&i,j\ge1,\;i\neq j.
  \end{cases}
\label{eq:dual-singval}
\end{equation}
Every singular value satisfies $|\alpha|^2\le\sigma_{ij}\le1$, with a uniform nonzero gap from zero. There is no severe compression: the partial trace rescales each mode by a factor between $|\alpha|^2$ and $1$.

The total spectral sum $S_\mathrm{total}=\sum_{i,j=0}^{d-1}\sigma_{ij}^{-1}$ is
\begin{align}
S_\mathrm{total}&=\underbrace{d}_{i=j}
+\underbrace{2(d-1)\frac{1}{|\alpha|}}_{(0,k)\text{ and }(k,0)}
+\underbrace{(d-1)(d-2)\frac{1}{|\alpha|^2}}_{i,j\ge1,\,i\neq j}\nonumber\\
&\sim\frac{d^2}{|\alpha|^2}\qquad(d\to\infty).
\label{eq:dual-Stotal}
\end{align}
The singular modes of $T$ are the rank-one operators $\tilde e_{ij}$ of \eqref{eq:dual-basisV} and their images. All of them have unit HS norm and unit trace norm. Taking $\tau$ below $\min_{ij}\sigma_{ij}=|\alpha|^2$ retains all $d^2$ modes, and Proposition~\ref{prop:overhead} gives
\begin{equation}
c(\cR_\tau)\le\mathfrak{C}_1(\tau)
=1+\frac{2}{d_B}\sum_{i,j=0}^{d-1}\frac{\|w_{ij}\|_1\|v_{ij}\|_1}{\sigma_{ij}}
=1+\frac{2S_{\mathrm{total}}}{d^2}\longrightarrow 1+\frac{2}{|\alpha|^2},
\label{eq:dual-pseudocost}
\end{equation}
where $v_{ij}$ and $w_{ij}$ are the left and right singular modes. Since $\cR_\tau$ acts as $T^{-1}$ on $\cW$ and its depolarizing correction vanishes on the blocks $Q_B^{(ij)}$, we have $\cR_\tau(Q_B^{(ij)})=Q_{BC}^{(ij)}$ for all $i,j$, so $\cR_\tau$ recovers $\rho_{ABC}$ exactly. We have
\begin{equation}
\nu_0(\ket{\psi_d})\le\log \mathfrak{C}_1(0^+)\le\frac{2}{|\alpha|^2},
\label{eq:dual-nu0}
\end{equation}
which is independent of $d$. By definition of the $\varepsilon$-approximate virtual Markovianity,
\begin{equation}
  \nu_\varepsilon(\ket{\psi_d})\le\nu_0(\ket{\psi_d})\le\frac{2}{|\alpha|^2}
\label{eq:dual-nu-ub}
\end{equation}
for every $\varepsilon>0$. Hence for fixed $|\alpha|\in(0,1)$ the quantity $\nu_\varepsilon(\ket{\psi_d})$ is bounded uniformly in $d$.

\subsubsection{CMI Asymptotics}

For a pure tripartite state, $I(A:C|B)=S(\rho_A)+S(\rho_C)-S(\rho_B)$. We compute each term.

\emph{Entropies of $\rho_A$ and $\rho_B$.} The Schmidt coefficients of $\ket{\psi_d}$ across the $A|BC$ cut are all equal to $1/d$ (the $d$ terms in \eqref{eq:dual-state} are orthogonal because the $\ket{v_i}$ are). Hence
\begin{equation}
  S(\rho_A)=\log d.
\label{eq:dual-SA}
\end{equation}

$\rho_B=\frac1d\sum_{i=0}^{d-1}\ket{v_i}\bra{v_i}$ is the uniform mixture of $d$ orthonormal vectors in a $d^2$-dimensional space. Its $d$ nonzero eigenvalues are each $1/d$, giving
\begin{equation}
  S(\rho_B)=\log d.
\label{eq:dual-SB}
\end{equation}

\emph{Entropy of $\rho_C$.} We have
\begin{equation}
  \rho_C=\frac1d\sum_{i=0}^{d-1}\ket{w_i}\bra{w_i}=\frac1d\Bigl[\ket0\bra0+\sum_{k=1}^{d-1}\bigl(|\alpha|^2\ket0\bra0+|\beta|^2\ket{k}\bra{k}\bigr)+\alpha\beta^*\ket0\bra{k}+\alpha^*\beta\ket{k}\bra0\Bigr].
\label{eq:dual-rhoC}
\end{equation}
In the basis $\{\ket0,\ket1,\dots,\ket{d-1}\}$, it takes the matrix form
\begin{equation}
  \rho_C=\frac1d\begin{bmatrix}
  1+(d-1)|\alpha|^2 & \alpha\beta^* & \alpha\beta^* & \cdots & \alpha\beta^*\\
  \alpha^*\beta & |\beta|^2 & 0 & \cdots & 0\\
  \alpha^*\beta & 0 & |\beta|^2 & \cdots & 0\\
  \vdots & \vdots & \vdots & \ddots & \vdots\\
  \alpha^*\beta & 0 & 0 & \cdots & |\beta|^2
\end{bmatrix}.
\end{equation}
This matrix is nontrivial only in the two-dimensional subspace $\Span\{\ket0,\ket{v}\}$ with $\ket{v}\equiv\frac{1}{\sqrt{d-1}}\sum_{k=1}^{d-1}\ket{k}$. This symmetry permits a block diagonalization. Choose an arbitrary orthonormal basis $\{\ket{v_1^\perp},\dots,\ket{v_{d-2}^\perp}\}$ of the orthogonal complement of $\ket{v}$ within $\Span\{\ket1,\dots,\ket{d-1}\}$. $\rho_C$ is block diagonal in this basis:
\begin{equation}
\begin{aligned}
  \rho_C&= \tilde\rho_C\oplus\frac{|\beta|^2}{d}I_{d-2},\\
  \tilde\rho_C&\equiv\frac1d
  \begin{bmatrix}
  1+(d-1)|\alpha|^2 & \alpha\beta^*\sqrt{d-1}\\[4pt]
  \alpha^*\beta\sqrt{d-1} & |\beta|^2
  \end{bmatrix}.
\end{aligned}
\end{equation}
Observe that $\Tr[\tilde\rho_C]=[d-(d-2)|\beta|^2]/d$ and $\det\tilde\rho_C=|\beta|^2/d^2$. For large $d$, the two eigenvalues of $\tilde\rho_C$ behave as
\begin{equation}
  \lambda_+=|\alpha|^2+O(d^{-1}),\quad
  \lambda_-=\frac{|\beta|^2}{d^2|\alpha|^2}+O(d^{-3}).
\label{eq:dual-lam}
\end{equation}
The remaining $(d-2)$-dimensional subspace carries the eigenvalue $|\beta|^2/d$ (degenerate). Hence
\begin{equation}
\begin{aligned}
  S(\rho_C)&=-\lambda_+\log\lambda_+-\lambda_-\log\lambda_--(d-2)\frac{|\beta|^2}{d}\log\frac{|\beta|^2}{d}\\
  &\sim-|\alpha|^2\log|\alpha|^2-|\beta|^2\log|\beta|^2+|\beta|^2\log d
\end{aligned}
\label{eq:dual-SC}
\end{equation}

\emph{CMI.} Assembling \eqref{eq:dual-SA}, \eqref{eq:dual-SB} and \eqref{eq:dual-SC}:
\begin{equation}
  I(A:C|B)_{\ket{\psi_d}}\sim|\beta|^2\log d\xrightarrow{d\to\infty}\infty.
\label{eq:dual-CMI}
\end{equation}
The CMI diverges logarithmically with $d$ whenever $|\beta|^2>0$ (i.e., $|\alpha|<1$), while the sampling overhead \eqref{eq:dual-nu0} stays bounded by a constant independent of $d$.

\subsection{Theorem on the Sampling Overhead}
\begin{theorem}[CMI does not determine the sampling overhead]\label{thm:CMI}
The following statements hold.
\begin{enumerate}
\item[(i)] For every continuous function $F:[0,\infty)\to[0,\infty)$ there exists a VQMC state $\rho$ with $\delta(\rho)=0$ such that
\begin{equation}\label{eq:cmi_upper}
\nu_0(\rho)>F\bigl(I(A:C|B)_\rho\bigr).
\end{equation}
Hence no continuous function of CMI serves as a universal upper bound on the exact overhead, nor on $\nu_\varepsilon$ in the near-floor regime, where $\nu_\varepsilon\to\nu_0$ as $\varepsilon\to0$.
\item[(ii)] For every function $G:[0,\infty)\times(0,\infty)\to[0,\infty)$ with $\lim_{x\to\infty}G(x,y)=\infty$ for all $y>0$, and every $\varepsilon>0$, there exists a VQMC state $\rho$ with $\delta(\rho)=0$ and
\begin{equation}\label{eq:cmi_lower}
\nu_\varepsilon(\rho)<G\bigl(I(A:C|B)_\rho,\;\varepsilon-\delta(\rho)\bigr).
\end{equation}
Hence no function of CMI is a universal lower bound on $\nu_\varepsilon$.
\end{enumerate}
\end{theorem}

\begin{proof}
\emph{Part (i).} Fix any $t>0$ and consider the family $\{\ket{\psi_d}\}_{d\ge3}$ defined in \eqref{eq:psid-expanded}, with the weight parameter $\eta(d)=t/\log d$ for all $d$ large enough that $\eta(d)<1$. Then $\eta\to0$ as $d\to\infty$. From \eqref{eq:CMI-final}, $I(A:C|B)_{\ket{\psi_d}}\longrightarrow t$. From \eqref{eq:delta-zero}, $\delta(\ket{\psi_d})=0$ for all $d$, so the state is a VQMC with $\varepsilon_{\min}=0$, and by \eqref{eq:nu-diverges} $\nu_0(\ket{\psi_d})\xrightarrow{d\to\infty}\infty$.
If the bound \eqref{eq:cmi_upper} held for all states, then for every $d$,
\begin{equation*}
\nu_0(\ket{\psi_d})\le F\bigl(I(A:C|B)_{\ket{\psi_d}}\bigr),
\end{equation*}
with the left-hand side diverging while the right-hand side converges to $F(t)$ and is therefore bounded. This is a contradiction, and $t>0$ was arbitrary. Hence no continuous function of CMI upper-bounds the exact overhead of VQMC states.

\emph{Part(ii).} Fix $\varepsilon>0$. Consider the family $\{\ket{\psi_d}\}_{d\ge2}$ defined in \eqref{eq:dual-state}. Eq.~\eqref{eq:dual-CMI} gives
\begin{equation}
  I(A:C|B)_{\ket{\psi_d}}\sim\frac12\log d\xrightarrow{d\to\infty}\infty.
\label{eq:dual-CMI-diverge}
\end{equation}
If the bound \eqref{eq:cmi_lower} were universally valid, we would have, for every $d$,
\begin{equation}
\nu_{\varepsilon}(\ket{\psi_d})\ge G\bigl(I(A:C|B)_{\ket{\psi_d}},\,\varepsilon-\delta(\ket{\psi_d})\bigr)=G\bigl(I(A:C|B)_{\ket{\psi_d}},\,\varepsilon\bigr)\to\infty.
\label{eq:dual-contra}
\end{equation}
This contradicts Eq.~\eqref{eq:dual-nu-ub}. Hence no function $G$ satisfying \eqref{eq:cmi_lower} serves as a universal lower bound.
\end{proof}

\begin{remark}
  One upper bound CMI does supply. By the Fawzi--Renner theorem every state admits a CPTP recovery with error at most $C_{\mathrm{FR}}\sqrt{I(A:C|B)}$. A CPTP map is a special HPTP map of overhead $1$, so $\nu_\varepsilon(\rho)=0$ whenever $\varepsilon\ge C_{\mathrm{FR}}\sqrt{I(A:C|B)}$. For general HPTP maps, this bound is vacuous in the high-accuracy regime $\varepsilon\to\varepsilon_{\min}$.
\end{remark}
\section{Conclusion}\label{sec:conclusion}
In conclusion, we have shown that virtual irrecoverability is of algebraic origin, quantified by the kernel failure measure $\delta$, which lower-bounds the error of every linear recovery map and thereby separates the states admitting exact virtual recovery from those that admit only approximate virtual recovery. For the latter we derived two-sided bounds on the sampling overhead in terms of the singular spectrum of the restricted partial trace, and gave two families of VQMCs showing that the CMI constrains this overhead in neither direction.

Several directions merit attention. First, beyond virtual recovery, what other operational tasks does the ghost information obstruct or enable, and can it be formalized as a quantum resource in its own right? Second, constructing explicit protocols for spin-squeezed \cite{Ma11PR}, N00N \cite{Boto00Sep}, and graph states \cite{Raussendorf01May} under photon loss would enable experimental validation in various quantum computing platforms. Third, extending the kernel failure analysis to quantum combs and higher-order operations \cite{Chiribella09PRA, Yang19Sep, Zhu24Jul} could reveal analogous obstructions in probabilistic error cancellation \cite{Temme17Nov, Endo18PRX} and quasi-probability classical shadows \cite{Jnane24PRXQuantum}.

\section*{Acknowledgments}

L.L. acknowledges support from the Fundamental and Interdisciplinary Disciplines Breakthrough Plan of the Ministry of Education of China (grant No. JYB2025XDXM115), the Ministry of Education of China Scientific Research Innovation Capability Support Project for Young Faculty (grant No. SRICSPYF-ZY2025009), the Tsinghua University Initiative Scientific Research Program, the National Key Research and Development Program of China (grant No. 2020YFA0715000), and the National Natural Science Foundation of China (grant No. 62075111); K.-M.X. acknowledges support from the Fundamental and Interdisciplinary Disciplines Breakthrough Plan of the Ministry of Education of China (grant Nos. JYB2025XDXM115 and JYB2025XDXM201) and the Beijing Science and Technology Planning Project (grant No. Z25110100040000).

\end{document}